\documentclass[letter,journal,final]{IEEEtran}  
\usepackage{arxiv}  
\usepackage{amsthm}  
\newcommand{\publicationversion}{Preprint
}

\usepackage{cite} 
\usepackage{amsmath,amssymb,amsfonts,nccmath}
\usepackage{algorithmic,algorithm}
\usepackage{graphicx}
\usepackage{textcomp}
\usepackage{color,soul}
\usepackage{tikz} 
\usepackage{subcaption}
\usepackage[colorlinks=true,linkcolor=blue, citecolor=blue, urlcolor=black]{hyperref}

\newtheorem{thm}{Theorem}
\newtheorem{lem}{Lemma}
\newtheorem{prop}{Proposition}

\newtheorem{assum}{Assumption}
\newtheorem{rem}{Remark}

\DeclareMathOperator{\tr}{tr}
\DeclareMathOperator{\dom}{dom}  
\newcommand{\T}{^{\top}} 
\newcommand{\ie}{\textit{i.e.}}
\allowdisplaybreaks

\def\BibTeX{{\rm B\kern-.05em{\sc i\kern-.025em b}\kern-.08em
		T\kern-.1667em\lower.7ex\hbox{E}\kern-.125emX}}
\begin{document}
	
	\title{Nonlinear Attitude Estimation Using Intermittent and Multi-Rate Vector Measurements}
	\author{
        \href{https://orcid.org/0000-0001-7498-5162}{\includegraphics[scale=0.06]{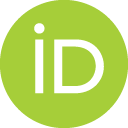}\hspace{0.5mm}} Miaomiao Wang \\
		Lakehead University\\
		Thunder Bay, ON P7B 5E1, Canada \\
		\texttt{mwang8@lakeheadu.ca}		
		\And			 	
		\href{https://orcid.org/0000-0002-9049-4651}{\includegraphics[scale=0.06]{orcid.png}\hspace{0.5mm}} Abdelhamid Tayebi \\
		Lakehead University\\
		Thunder Bay, ON P7B 5E1, Canada \\
		{\tt atayebi@lakeheadu.ca} 
	}%
\twocolumn[ 
\begin{@twocolumnfalse}
\maketitle 
	
	\begin{abstract}
		This paper considers the problem of nonlinear attitude estimation for a rigid body system using intermittent and multi-rate inertial vector measurements as well as continuous (high-rate) angular velocity measurements. Two types of hybrid attitude observers on Lie group $SO(3)$  are proposed. First, we propose a hybrid attitude observer where almost global asymptotic stability is guaranteed using the notion of almost global input-to-state stability on manifolds. Thereafter, this hybrid attitude observer is extended by introducing a switching mechanism to achieve global asymptotic stability. Both simulation and experimental results are presented to illustrate the performance of the proposed hybrid observers.
	\end{abstract}
	
\end{@twocolumnfalse}]
\copyrightfootnote{This work has been accepted for publication in IEEE TAC.}

	\section{Introduction}
	The algorithms used for the determination of the attitude (or orientation) of a rigid body system are instrumental in many applications related to robotics, aerospace and marine engineering. Since the attitude is not directly measurable from any sensor, it can be obtained through the integration of the angular velocity or determined from body-frame observations of at least two non-collinear vectors known in the inertial frame. The latter, known as Wahba's problem \cite{Wahba1965}, relies on vector observations obtained from different types of sensors such as low-cost inertial measurement unit (IMU) sensors (including an accelerometer, a gyroscope and a magnetometer), or sophisticated sensors such as sun sensors and star trackers. However, although simple, both approaches do not perform well in the presence of measurement bias and noise. This motivated the use of Kalman-type filters leading to dynamic attitude estimation algorithms (see the survey paper \cite{Crassidis2007survey,shuster1979}). Although successfully implemented in many practical applications, these Kalman-based dynamic estimation techniques rely on local linearizations/approximations and lack stability guarantees in the global sense.

	Recently, a class of geometric nonlinear attitude observers have made their appearances in the literature, for instance, nonlinear complementary filters on the Special Orthogonal group $SO(3)$ \cite{mahony2008nonlinear,bonnabel2008symmetry,hua2013implementation ,Berkane_Tay_TAC2017} and the invariant extended Kalman filter \cite{barrau2015intrinsic}. Unlike the classical attitude observers/filters, the geometric observers take into account the topological properties of the group $SO(3)$ and can provide almost global asymptotic stability (AGAS) guarantees, \ie, the estimated attitude converges asymptotically to the actual one from almost all initial conditions except from a set of zero Lebesgue measure. It is important to point out that $SO(3)$ is a compact odd-dimensional manifold without boundary, which is not homeomorphic to the vector space $\mathbb{R}^n$, and as such, there is no smooth vector field with a global attractor on $SO(3)$. Therefore,  the strongest result one can achieve via these time-invariant smooth observers is AGAS (see, for instance \cite{Koditschek88}). Motivated by the work in  \cite{mayhew2011hybrid,mayhew2011synergistic,mayhew2013synergistic} and the framework of hybrid dynamical systems in \cite{goebel2009hybrid,goebel2012hybrid}, hybrid attitude observers with global asymptotic stability guarantees have been proposed in \cite{wu2015globally,berkane2017hybrid}. 
	The idea of these hybrid observers, with global asymptotic stability guarantees, has been extended to other more complicated state estimation problems on matrix Lie groups such as hybrid pose observers on group $SE(3)$ \cite{wang2019hybrid} and hybrid state observers for inertial navigation systems on group $SE_2(3)$ \cite{wang2020hybrid}.

	From the practical point of view, attitude estimation often involves different types of sensors with different sampling rates. For instance, in the problem of vision-aided attitude estimation, the measurements from a vision sensor are obtained at rates as low as 20 Hz, which is much lower than the rates of the IMU measurements (up to 1000 Hz). However, most of the existing attitude observers in the literature are designed based on continuous output measurements, for instance, \cite{mahony2008nonlinear,hua2013implementation,zlotnik2016nonlinear,berkane2017hybrid,cabecinhas2019integrated}. One common way to implement these continuous observers, using intermittent output measurements, is to apply the zero-order-hold (ZOH) method. Unfortunately, the stability and convergence guarantees are not necessarily preserved under this practical ad-hoc setup. In this context, some state estimation schemes on $\mathbb{R}^n$ relying on intermittent output measurements have been proposed, for instance, in \cite{ferrante2016state, sferlazza2018time, ferrante2021observer, landicheff2022continuous, chen2022variable}. Other attitude estimation schemes, using discrete inertial vector measurements, have been developed on the Lie group $SO(3)$, such as the discrete-time attitude observers in \cite{laila2011discrete, izadi2015discrete, barrau2015intrinsic, bhatt2020rigid}, and the continuous-discrete attitude observers in \cite{khosravian2015recursive, berkane2019attitude}. The latter category assumes continuous (high-rate) measurements of the angular velocity and intermittent measurements of inertial vectors with different sampling rates. A predictor-observer approach has been proposed in \cite{khosravian2015recursive} based on a cascade combination of an output predictor and a continuous attitude observer. 
	In \cite{berkane2019attitude}, the authors developed a (non-smooth) predict-update hybrid estimation scheme, where the estimated attitude is continuously updated by integrating the attitude kinematics using the continuous angular velocity measurements and discretely updated through jumps upon the arrival of the intermittent vector measurements. Both results in \cite{khosravian2015recursive} and \cite{berkane2019attitude}  only guarantee AGAS due to the topological obstruction on $SO(3)$ and the nature of the intermittent inertial vector measurements.

	In this paper, we consider the problem of attitude estimation using continuous (high-rate) angular velocity and intermittent inertial vector measurements with multiple sampling rates. We first propose a hybrid nonlinear observer on manifold $SO(3)\times \mathbb{R}^{n}$ endowed with AGAS guarantees using the notion of almost global input-to-state stability (ISS) on manifolds presented in \cite{angeli2004almost}. In this hybrid observer, the estimated states are continuously updated through integration using the continuous angular velocity measurements and discreetly updated upon the arrival of the intermittent vector measurements. To achieve global asymptotic stability (GAS), we propose a new hybrid observer with a switching mechanism motivated from \cite{wang2021hybrid}.
	The contribution of this paper can be summarized as follows:
	\begin{itemize}
		\item [1)] \textit{Multi-rate vector measurements:} The attitude estimation observers proposed in this work can handle intermittent vector measurements with different sampling rates (\ie, asynchronously-intermittent measurements) where not all the measurements are received at the same time. For instance, in many practical applications, the sampling rate of the IMU  is much higher than that of global positioning systems (GPS) and vision sensors. This is a key difference with respect to most of the existing attitude observers assuming that the vector measurements are continuous or discrete with the same sampling rate \cite{mahony2008nonlinear,hua2013implementation,berkane2017hybrid,laila2011discrete,izadi2015discrete,barrau2016invariant,wang2021cdc}. Our simulation results validate that the convergence is not guaranteed when implementing continuous attitude observers (for instance, the complementary filter \cite{mahony2008nonlinear}) with ZOH method.
		
		\item [2)] \textit{Smooth attitude estimation:} The observers proposed in this paper have a similar continuous-discrete structure as \cite{barrau2016invariant,berkane2019attitude}, while the estimated attitude from our hybrid observers is continuous without any additional smoothing algorithm as in \cite{berkane2019attitude}. The fact that our proposed hybrid observers generate continuous estimates of the attitude makes it suitable for practical applications involving observer-controller implementations.
		
		\item [3)] \textit{Global asymptotic stability:} In contrast to the observers proposed in this paper, the existing attitude observers can only guarantee local or almost global asymptotic stability when dealing with intermittent vector measurements, for instance, \cite{khosravian2015recursive,barrau2016invariant,berkane2019attitude,wang2021cdc}. To the best of our knowledge, this is the first work dealing with nonlinear smooth attitude estimation with GAS guarantees in terms of intermittent and multi-rate vector measurements.
		
		\item [4)] \textit{Experimental validation:} In this paper, our proposed hybrid attitude observer has been experimentally validated using the measurements obtained from an IMU and an RGB-D camera, and compared against some state-of-the-art attitude estimation/determination algorithms. 
		
	\end{itemize}
	
	The remainder of this paper is organized as
	follows: Section \ref{sec:backgroud} provides the preliminary materials that will be used
	throughout this paper. In Section \ref{sec:problem}, we formulate our attitude estimation problem in terms of intermittent vector measurements. In Section \ref{sec:hybrod_obs_AGAS}, a new hybrid attitude observer with AGAS guarantees is proposed. In
	Sections \ref{sec:hybrod_obs_GAS}, we propose a new hybrid attitude observer with GAS guarantees.  Simulation and experimental results are presented in Section \ref{sec:simulation_experiments} to illustrate the performance of the proposed observers.

	\section{Preliminary Material} \label{sec:backgroud}
	\subsection{Notations and Definitions}
	The sets of real, non-negative real, and non-zero natural numbers are denoted by $\mathbb{R}$, $\mathbb{R}_{\geq 0}$, and $\mathbb{N}$, respectively. We denote by $\mathbb{R}^n$ the $n$-dimensional Euclidean space and $\mathbb{S}^{n-1}$ the set of unit vectors in $\mathbb{R}^{n}$.  The Euclidean norm of a vector $x\in \mathbb{R}^n$ is defined as $\|x\| = \sqrt{x\T x}$. Let $I_n$ denote the  $n$-by-$n$ identity matrix and  $0_{n \times m}$ denote the $n$-by-$m$ zero matrix.
	For a given matrix $A\in \mathbb{R}^{n\times n}$, we define $(\lambda_i^A,v_i^A)$ as its $i$-th pair of eigenvalue and eigenvector, and $\mathcal{E}(A):=\{v\in \mathbb{R}^n: v = v_i^A/\|v_i^A\|, Av_i^A = \lambda_i^A v_i^A  \}$ as the set of all unit eigenvectors of $A$.
	Given two matrices $A, B\in \mathbb{R}^{m\times n}$, their Euclidean inner product is defined as $\langle\langle A, B \rangle\rangle = \tr(A\T B)$ where $\tr(\cdot)$ represents the trace of a square matrix, and the Frobenius norm of matrix $A$ is defined as $\|A\|_F = \sqrt{\langle\langle A, A\rangle\rangle} = \sqrt{\tr(A\T A)}$. For each vector $x=[x_1,x_2,x_3]\T \in \mathbb{R}^3$, we define $x^\times$ as a skew-symmetric matrix given by
	\[
	x^\times = \begin{bmatrix}
		0    & -x_3 & x_2  \\
		x_3  & 0    & -x_1 \\
		-x_2 & x_1  & 0
	\end{bmatrix}.
	\]
	For a matrix $A=[a_{ij}]_{1\leq i,j \leq 3} \in \mathbb{R}^{3\times 3}$, we define
	$
	\psi (A) :=
	\frac{1}{2}[
	a_{32} - a_{23},
	a_{13} - a_{31},
	a_{21} - a_{12}
	]\T.
	$
	For any $A\in \mathbb{R}^{3\times 3}, x\in \mathbb{R}^3$, one can verify that
	$
	\langle\langle A, x^\times\rangle\rangle
	= 2x\T \psi (A)
	$. 
	We denote the 3-dimensional  \textit{Special Orthogonal group}  by
	$
	SO(3): = \left\{R\in \mathbb{R}^{3\times 3}| R\T R = I_3, \det(R) = +1\right\}
	$
	and its \textit{Lie algebra} by
	$
	\mathfrak{so}(3) := \left\{\Omega\in \mathbb{R}^{3\times 3}| \Omega\T=-\Omega \right\}.
	$ 
	Let the map $\mathcal{R}_a: \mathbb{R} \times \mathbb{S}^2\to SO(3)$ represent the well-known \textit{angle-axis parameterization} of the attitude, which is given by
	$$
	\mathcal{R}_a(\theta,u) :=  I_3 + \sin(\theta) u^\times + (1-\cos(\theta))(u^\times)^2
	$$
	with $u\in \mathbb{S}^2$ indicating the direction of an axis of rotation and $\theta\in \mathbb{R}$ describing the angle of the rotation about the axis. 

	\subsection{Hybrid Systems Framework}
	Consider a  smooth manifold $\mathcal{M}$ embedded in vector space $\mathbb{R}^{n}$. 
	Let $T_x \mathcal{M}$ denote the \textit{tangent space}\footnote{
		For each $x\in \mathcal{M}$, its tangent space $T_{x}\mathcal{M}$ is defined as the set of all the tangent vectors to $\mathcal{M}$ at $x$, where the tangent vector at $x\in \mathcal{M}$ is defined as $\dot{\gamma}(0) = \frac{d \gamma(t)}{dt}|_{t=0}$, with $\gamma$ being a differentiable curve defined as $\gamma:\mathbb{I} \to \mathcal{M}$ satisfying $\gamma(0) = x$, with $\mathbb{I} \subset \mathbb{R}$ being an open interval containing zero in its interior.
	} %
	to $\mathcal{M}$ at $x$, and $T \mathcal{M}:=\bigcup_{x\in \mathcal{M}} T_x \mathcal{M}$ 
	denote the \textit{tangent bundle} of manifold $\mathcal{M}$. A general model of a hybrid system is given as \cite{goebel2012hybrid}:
	\begin{equation}\mathcal{H}:
		\begin{cases}
			\dot{x} ~~= F(x), & \quad  x \in \mathcal{F} \\
			x^{+} \in G(x),   & \quad  x \in \mathcal{J}
		\end{cases} \label{eqn:hybrid_system}
	\end{equation}
	where $x\in \mathcal{M}$ denotes the state, $x^+$ denotes the state after an instantaneous jump, the \textit{flow map} $F: \mathcal{M} \to T \mathcal{M}$ describes the continuous flow of $x$ on the \textit{flow set} $\mathcal{F} \subseteq \mathcal{M}$, and the \textit{jump map} $G: \mathcal{M}\rightrightarrows  \mathcal{M}$ (a set-valued mapping from $\mathcal{M}$  to $\mathcal{M}$) describes the discrete jump of $x$ on the \textit{jump set} $\mathcal{J} \subseteq \mathcal{M}$.
	A solution $x$ to $\mathcal{H}$ is parameterized by $(t, j) \in \mathbb{R}_{\geq 0} \times \mathbb{N}$, where $t$ denotes the amount of time that has passed and $j$ denotes the number of discrete jumps that have occurred.  A subset $\dom x \subset \mathbb{R}_{ \geq 0} \times \mathbb{N}$ is a \textit{hybrid time domain} if for every $(T,J)\in \dom x$, the set, denoted by $\dom x \bigcap ([0,T]\times \{0,1,\dots,J\})$, is a union of finite intervals  of the form $\bigcup_{j=0}^{J} ([t_j,t_{j+1}] \times \{j\})$ with  a time sequence  $0=t_0 \leq t_1 \leq \cdots \leq t_{J+1}$.
	A solution $x$ to $\mathcal{H}$ is said to be  \textit{maximal} if
	it cannot be extended by flowing nor jumping, and \textit{complete} if its domain $\dom x$ is unbounded. Let $|x|_{\mathcal{A}}$  denote the distance of a point $x$ to a closed set $\mathcal{A} \subset \mathcal{M}$, 
	and then the set $\mathcal{A}$ is said to be:
	\textit{stable} for $\mathcal{H}$ if for each $\epsilon>0$ there exists $\delta>0$ such that each maximal solution $x$ to $\mathcal{H}$ with $|x(0,0)|_{\mathcal{A}} \leq \delta$ satisfies $|x(t,j)|_{\mathcal{A}} \leq \epsilon$ for all $(t,j)\in \dom x$; \textit{globally attractive} for $\mathcal{H}$ if  every maximal solution $x$ to $\mathcal{H}$  is complete and satisfies $\lim_{t+j\to \infty}|x(t,j)|_{\mathcal{A}} = 0$ for all $(t,j)\in \dom x$; \textit{globally asymptotically stable} if it is both stable and globally attractive for $\mathcal{H}$.
	Moreover, the $\mathcal{A}$ is said to be   \textit{exponentially stable} for   $\mathcal{H}$ if there exist $\kappa, \lambda>0$ such that, every maximal solution  $x$  to $\mathcal{H}$ is complete and satisfies  $|x(t,j)|_{\mathcal{A}} \leq \kappa e^{-\lambda(t+j)}|x(0,0)|_{\mathcal{A}}$ for all $(t,j)\in \dom x$ \cite{teel2013lyapunov}.
	We refer the reader to \cite{goebel2009hybrid,goebel2012hybrid} and references therein for more details on hybrid dynamical systems.
	
	\section{Problem Statement}\label{sec:problem}
	Let $\{\mathcal{I}\}$ be an inertial frame and $\{\mathcal{B}\}$ be a body-fixed frame attached to the center of mass of a rigid body.
	Consider the attitude kinematics  for a rigid body on $SO(3)$ as
	\begin{align}
		\dot{R} & = R\omega^\times \label{eqn:dynamics_R}
	\end{align}
	where the rotation matrix $R\in SO(3)$ denotes the attitude (orientation) of the body-fixed frame $\{\mathcal{B}\}$  with respect to the inertial frame $\{\mathcal{I}\}$, and the vector $\omega\in \mathbb{R}^3$ denotes the angular velocity of the rigid body expressed in body-fixed frame $\{\mathcal{B}\}$. In practice, the angular velocity can be obtained from the gyroscopes at a very high rate. Therefore, we assume that the body-fixed frame angular velocity $\omega$ is continuously measurable.

	Consider a family of $N\geq 2$ constant vectors known in the inertial frame, namely inertial vectors, denoted by $r_i\in \mathbb{R}^3$ for all $i\in  \mathbb{I}:= \{1,2,\cdots,N\} $. The measurements of the inertial vectors expressed in the body-fixed frame are modelled as
	\begin{equation}
		b_i  = R \T r_i   \label{eqn:def_b_i}
	\end{equation}
	for all $i\in   \mathbb{I}$ and satisfy the following assumptions:
	\begin{assum}\label{assum:non-collinear}
		There exist at least two non-collinear vectors among the $N\geq 2$ inertial vectors.
	\end{assum}
	
	\begin{assum}\label{assum:intermittent}
		For each inertial vector $r_i, i\in \mathbb{I}$, the time sequence $\{t_k^i\}_{ k\in \mathbb{N}}$ of its measurements is strictly increasing, and there exist two constants $0<T_m^i \leq T_M^i < \infty$ such that $0 \leq t_{1}^i  \leq T_M^i$ and $T_m^i \leq t_{k+1}^i-t_k^i \leq T_M^i $ for all $ 1\leq k\in \mathbb{N} $.
	\end{assum}
	\begin{rem}		
		Note that Assumption \ref{assum:non-collinear} is commonly used, for observability purposes, for the development of attitude estimation schemes, see for instance \cite{mahony2008nonlinear,hua2013implementation,khosravian2015recursive,berkane2017hybrid}. 
         Assumption \ref{assum:intermittent} implies that the measurements of the inertial vectors can be irregular and have different sampling periods. Note also that the sampling is periodic with a regular sampling period if $T_m^i=T_M^i$ for all $i \in \mathbb{I}$. In practice, the inertial vector measurements can be obtained from different sensors (such as magnetometers, accelerometers, vision sensors, star trackers and sun sensors), asynchronously with different and time-varying sampling periods. It is important to point out that, as it is going to be shown later, these lower and upper bounds are not required for the observer design and stability analysis of our proposed hybrid observers.
	\end{rem}
	
	Let $\hat{R}\in SO(3)$ denote the estimate of the attitude $R$ and $\tilde{R}:= R\hat{R}\T \in SO(3)$ denote the attitude estimation error. The objective of this work is to develop a hybrid attitude estimation scheme on $SO(3)$ for system \eqref{eqn:dynamics_R} guaranteeing that the attitude estimation error $\tilde{R}$ converges to $I_3$, 
    using continuous angular velocity measurements and intermittent multi-rate body-frame vector measurements under Assumptions \ref{assum:non-collinear} and \ref{assum:intermittent}.
	
	\section{Hybrid Observer with AGAS Guarantees} \label{sec:hybrod_obs_AGAS}
	\subsection{Observer Design}
	In this section, we will design a hybrid estimation scheme to handle intermittent and multi-rate vector measurements. Let $\hat{R}\in SO(3)$ denote the estimate of the attitude $R$, and $\hat{r}_i\in \mathbb{R}^3$ denote the estimate of the vector $\hat{R}b_i$ corresponding to the $i$-th inertial vector $r_i$. Motivated by \cite{wang2021cdc}, we propose the following hybrid observer on manifold $SO(3)\times \mathbb{R}^{3N}$:
	\begin{subequations}\label{eqn:hybrid_observer0}
		\begin{align}
			& \dot{\hat{R}} = \hat{R}(\omega + k_o\hat{R}\T \sigma_R)^\times \label{eqn:hybrid_observer0-1}                                                       \\
			& \begin{cases}
				\dot{\hat{r}}_i ~ = k_o \sigma_R^\times \hat{r}_i,              &  t \neq t_k^i,~    k\in \mathbb{N} \\
				\hat{r}_i^+   =  \hat{r}_i +    k_{r} (\hat{R}{b}_i-\hat{r}_i), &  t =t_k^i,~  k\in \mathbb{N}
			\end{cases}
			\label{eqn:hybrid_observer0-2}
		\end{align}
	\end{subequations}
	with scalar gains $k_o>0$ and $0<k_r<1$. The innovation term $\sigma_R$ is designed as
	\begin{align}
		\sigma_R & =   \sum_{i=1}^N \rho_i     \hat{r}_i  \times   r_i \label{eqn:def_sigma_R_0} 
	\end{align}
	where $\rho_i>0$ for all $i\in \mathbb{I}$. The structure of our proposed hybrid attitude observer \eqref{eqn:hybrid_observer0} is given in Fig. \ref{fig:hybrid_observer0}.
	\begin{figure}[!ht]
		\centering
		\tikzset{every picture/.style={line width=0.75pt}} 
		
		\begin{tikzpicture}[x=0.43pt,y=0.48pt,yscale=-1,xscale=1] \small
			
			\draw   (100,126) -- (242,126) -- (242,171.27) -- (100,171.27) -- cycle ;
			\draw   (330,129.27) -- (493,129.27) -- (493,175.27) -- (330,175.27) -- cycle ;
			\draw    (243,148.36) -- (328,148.28) ;
			\draw [shift={(330,148.28)}, rotate = 539.94] [color={rgb, 255:red, 0; green, 0; blue, 0 }  ][line width=0.75]    (10.93,-4.9) .. controls (6.95,-2.3) and (3.31,-0.67) .. (0,0) .. controls (3.31,0.67) and (6.95,2.3) .. (10.93,4.9)   ;
			\draw    (492.4,152.15) -- (559.75,152.27) ;
			\draw [shift={(561.75,152.27)}, rotate = 180.1] [color={rgb, 255:red, 0; green, 0; blue, 0 }  ][line width=0.75]    (10.93,-4.9) .. controls (6.95,-2.3) and (3.31,-0.67) .. (0,0) .. controls (3.31,0.67) and (6.95,2.3) .. (10.93,4.9)   ;
			\draw    (532,152) -- (532.2,217) ;
			\draw    (170.6,217) -- (532.2,217) ;
			\draw    (170.6,217) -- (170.6,173.8) ;
			\draw [shift={(170.6,171.8)}, rotate = 450] [color={rgb, 255:red, 0; green, 0; blue, 0 }  ][line width=0.75]    (10.93,-4.9) .. controls (6.95,-2.3) and (3.31,-0.67) .. (0,0) .. controls (3.31,0.67) and (6.95,2.3) .. (10.93,4.9)   ;
			\draw  [dash pattern={on 4.5pt off 4.5pt}]  (14,145.6) -- (100,145.79) ;
			\draw [shift={(100.2,145.8)}, rotate = 180.17] [color={rgb, 255:red, 0; green, 0; blue, 0 }  ][line width=0.75]    (10.93,-4.9) .. controls (6.95,-2.3) and (3.31,-0.67) .. (0,0) .. controls (3.31,0.67) and (6.95,2.3) .. (10.93,4.9)   ;
			\draw    (413,105) -- (413,125.86) ;
			\draw [shift={(413,127.86)}, rotate = 271.79] [color={rgb, 255:red, 0; green, 0; blue, 0 }  ][line width=0.75]    (10.93,-4.9) .. controls (6.95,-2.3) and (3.31,-0.67) .. (0,0) .. controls (3.31,0.67) and (6.95,2.3) .. (10.93,4.9)   ;

			\draw (102,129) node [anchor=north west][inner sep=0.75pt]   [align=left] {Hybrid vector\\ estimation \eqref{eqn:hybrid_observer0-2}};
			\draw (332,132.27) node [anchor=north west][inner sep=0.75pt]   [align=left] {Attitude\\ estimation  \eqref{eqn:hybrid_observer0-1}};
			\draw (14,123) node [anchor=north west][inner sep=0.75pt]   [align=left] {$b_1,...,b_N$};
			\draw (243,126) node [anchor=north west][inner sep=0.75pt]   [align=left] {$\hat{r}_1,...,\hat{r}_N$};
			\draw (498,128) node [anchor=north west][inner sep=0.75pt]   [align=left] {$\hat{R}$};
			\draw (420,110) node [anchor=north west][inner sep=0.75pt]   [align=left] {$\omega$};
		\end{tikzpicture}
		\caption{The architecture of the proposed hybrid observer \eqref{eqn:hybrid_observer0}. 
		}\label{fig:hybrid_observer0}
	\end{figure}
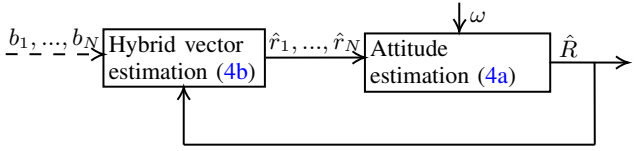
	
	The dynamics of the attitude estimate $\hat{R}$ in \eqref{eqn:hybrid_observer0-1} are designed through a continuous integration using the angular velocity and the innovation term $\sigma_R$. Note that the nonstandard innovation term $\sigma_R$ designed in \eqref{eqn:def_sigma_R_0} relies on the inertial vectors $r_i$ and the vector estimates $\hat{r}_i$. Typically, the innovation term for attitude estimation can be designed as $\sigma_R =  \sum_{i=1}^N \rho_i \hat{R} b_i \times r_i $ if the body-frame vector measurements $b_i, i\in \mathbb{I}$ are continuously available as in \cite{mahony2008nonlinear}.  However, the estimation problem considered in this work is quite challenging since we do not have the continuous vector measurements that allow the construction of a continuous innovation term. To overcome this challenge, we introduce the vector estimate $\hat{r}_i$ and design the innovation term $\sigma_R$ as in \eqref{eqn:def_sigma_R_0} which relies on the vector estimate $\hat{r}_i$ instead of the vector measurement $b_i$. Additionally, the hybrid dynamics for the vector estimate $\hat{r}_i$ are given in \eqref{eqn:hybrid_observer0-2} to ensure that the vector $\hat{r}_i$ tends exponentially to  $\hat{R} b_i$. As it is going to be shown later, the vector $\hat{r}_i$ is continuously updated through integration between vector measurements using the innovation term $\sigma_R$ to guarantee that the estimation error $\hat{r}_i - \hat{R}b_i$ is non-increasing, and discreetly updated upon the arrival of the intermittent vector measurements to guarantee an exponential decrease of the estimation error.

	It is clear from \eqref{eqn:hybrid_observer0-1} that the solutions of the attitude estimate $\hat{R}(t)$ are continuous for all $t\geq 0$ (not necessarily differentiable due to the discrete jumps of $\hat{r}_i$ in the hybrid dynamics \eqref{eqn:hybrid_observer0-2}), which is a key difference with respect to the existing work in \cite{barrau2016invariant,berkane2019attitude}. It is important to point out that continuous/smooth attitude estimates are instrumental
	for practical applications involving observer-controller implementations.
	A similar idea of vector prediction has been proposed in \cite{khosravian2015recursive} relying on the forward integration of the angular velocity on $SO(3)$. The main advantage of our work is that no (large) memory or buffer is required to compute the predicted vector measurements at each time.
	
	\subsection{Stability Analysis}
	To capture the behaviour of the event-triggered system \eqref{eqn:hybrid_observer0-2}, a virtual timer $\tau_i$ for the $i$-th vector measurement is considered, whose hybrid dynamics are modelled as:
	\begin{equation}
		\begin{cases}
			\dot{\tau}_i ~~= -1,         & \tau_i\in [0,T_M^i] \\
			\tau_i^+ \in [T_m^i, T_M^i], & \tau_i\in \{0\}
		\end{cases} \label{eqn:tau}
	\end{equation}
	with $\tau_i(0,0)\in [0,T_M^i]$. The hybrid dynamics of the virtual timers are motivated from \cite{ferrante2016state} and have been considered in \cite{wang2020nonlinear,wang2021cdc}. For each vector $i\in \mathbb{I}$, the virtual timer  $\tau_i$
	decreases to zero continuously, and upon reaching zero it is automatically reset to a value, between $T_m^i$ and $T_M^i$, which represents the arrival time of the next measurement of the $i$-th vector. With these additional states $\tau_i, i\in \mathbb{I}$, the time-driven sampling events can be described as state-driven events, which results in an autonomous hybrid closed-loop system as it is going to be shown later.

	Consider the geometric attitude estimation error $\tilde{R} := R\hat{R}\T$ and the vector estimation errors $\tilde{r}_i :=  r_i- R\hat{R}\T\hat{r}_i = r_i - \tilde{R}\hat{r}_i$ for all $i\in \mathbb{I}$. For the sake of simplicity, we define $\tilde{r} := [\tilde{r}_1\T,\tilde{r}_2\T,\dots,\tilde{r}_N\T]\T \in \mathbb{R}^{3N}$ and $\tau := [\tau_1,\dots,\tau_N]\T \in \mathbb{R}^{N}$. Then, the innovation term $\sigma_R$ defined in  \eqref{eqn:def_sigma_R_0} can be rewritten in terms of the estimation errors as
	\begin{align}
		\sigma_R
		& = - \sum_{i=1}^N \rho_i  r_i^\times  \tilde{R}\T (r_i-\tilde{r}_i) \nonumber  \\
		& = -   \sum_{i=1}^N \rho_i  (  r_i^\times   \tilde{R}\T r_i   - r_i^\times   \tilde{R}\T  \tilde{r}_i  )
		\nonumber \\
		& =  \psi(A\tilde{R} ) +  \Gamma(\tilde{R}) \tilde{r}   \label{eqn:error_sigma_R0}
	\end{align}
	where $\Gamma(\tilde{R}): = \begin{bmatrix}
		\rho_1 r_1^\times \tilde{R}\T,   \rho_2   r_2^\times \tilde{R}\T, \dots, \rho_N   r_N^\times \tilde{R}\T
	\end{bmatrix} $ and  we have made use of the fact $\psi(A\tilde{R} ) = - \sum_{i=1}^N \rho_i r_i^\times \tilde{R}\T r_i$ with
	\begin{equation}
		A= \sum_{i=1}^N \rho_i r_i r_i\T \in \mathbb{R}^{3\times 3}. \label{eqn:def_A}
	\end{equation}
	Hence, from \eqref{eqn:dynamics_R}, \eqref{eqn:hybrid_observer0}, \eqref{eqn:tau} and \eqref{eqn:error_sigma_R0}, the error dynamics of $\tilde{R}$ and $\tilde{r}_i $ are given as follows:
	\begin{subequations}\label{eqn:hybrid_observer_I}
		\begin{align}
			& \dot{\tilde{R}}= \tilde{R} (  -k_o\psi(A\tilde{R}) - k_o \Gamma(\tilde{R}) \tilde{r}  ) ^\times  \label{eqn:dot_tildeR} \\
			& \begin{cases}
				\dot{\tilde{r}}_i ~ = 0,               &  \tau_i\in [0,T_M^i] \\
				\tilde{r}_i^+   = (1-k_r) \tilde{r}_i, &  \tau_i\in \{0\}
			\end{cases}  \label{eqn:hybrid_tilde_r}
		\end{align}
	\end{subequations}
	
	It is clear that the overall closed-loop system \eqref{eqn:tau} and \eqref{eqn:hybrid_observer_I} is autonomous. As pointed out in \cite{mahony2008nonlinear}, it is always possible to tune the scalar weights $\rho_i>0, i\in \mathbb{I}$ such that the matrix $A=A\T$ defined in \eqref{eqn:def_A} is positive definite if there are $N\geq 3$ linearly independent inertial vectors. In the case where only two inertial vectors are non-collinear, for instance, $r_1$ and $r_2$, one can always construct an additional inertial vector $r_{N+1} = r_1 \times r_2$ and its corresponding body-frame vector $b_{N+1} = b_1 \times b_2$. Therefore, under Assumption \ref{assum:non-collinear}, it is reasonable to assume that the matrix $A$ defined in \eqref{eqn:def_A} is positive definite with three distinct eigenvalues $0<\lambda_1^A < \lambda_2^A < \lambda_3^A$.
	
	\begin{prop}\label{prop:AISS-SO3}
		Consider the following system on $SO(3)$:
		\begin{align}
			\dot{\tilde{R}}= \tilde{R} (  -k_o\psi(A\tilde{R}) +  \bar{\Gamma}(\hat{R}) u  ) ^\times  \label{eqn:aiss}
		\end{align}
		where $\tilde{R}, \hat{R}\in SO(3)$, $k_o>0$, $\bar{\Gamma}:SO(3) \to \mathbb{R}^{3\times m}$ and $u\in \mathcal{D}_u\subset \mathbb{R}^m$ with $\mathcal{D}_u $ being  a closed subset of $\mathbb{R}^m$. 
		Suppose that $A$ is positive definite with three distinct eigenvalues, and $\|\bar{\Gamma}(X)\|_F\leq c_{\bar{\Gamma}}$ for all $X\in SO(3)$ with a constant $c_{\bar{\Gamma}}>0$.
		Then, system \eqref{eqn:aiss} is almost globally input-to-state stable with respect to the equilibrium  $I_3$ and input $u$.
	\end{prop}
	\begin{rem}
		The proof of Proposition \ref{prop:AISS-SO3} can be conducted by following the same steps as in the proof of \cite[Proposition 1]{wang2021CDC_arxiv}, and is therefore omitted here. The proof relies on the results in \cite[Proposition 2]{angeli2010stability}, along with the facts that system \eqref{eqn:aiss} is AGAS with zero input (\ie, $u\equiv 0$) and has finite exponentially unstable isolated equilibria since that $A$ is positive definite with three distinct eigenvalues. A similar result on almost global ISS of system \eqref{eqn:aiss}, with $A = I_3$ and some high gain $k_o$ depending on the bound of the input $u$, can be found in \cite{vasconcelos2011combination} using a combination of Lyapunov and density functions.
	\end{rem}

	\begin{prop}\label{prop:tilde_r}
		Consider the hybrid system \eqref{eqn:tau} and \eqref{eqn:hybrid_tilde_r} with $0<k_r<1$.
		Then, the vector estimation error $\tilde{r}_i, i\in \mathbb{I}$ converges (globally exponentially) to zero, \ie, there exist constants $\alpha, \lambda>0$ such that
		\begin{equation}
			\|\tilde{r}_i(t,j)\|^2 \leq \alpha  e^{- \lambda(t+j)}  \|\tilde{r}_i(0,0)\|^2  \label{eqn:tilde_r_exp}
		\end{equation}
		for all $(t,j)\in \dom (\tilde{r}_i, \tau_i)$. 
	\end{prop}
	\begin{proof}
		See Appendix \ref{sec:tilde_r}.
	\end{proof}
	\begin{rem}
		From Proposition \ref{prop:tilde_r}, it implies that  $ \|\tilde{r}(t,j)\|^2 \leq \alpha  e^{- \lambda t}  \|\tilde{r}(0,0)\|^2$ for all $(t,j)\in \dom (\tilde{r},\tau) $ with some constants $\alpha, \lambda>0$. It is important to point out that the global convergence of the vector estimation error $\tilde{r}$ is independent from the convergence of the attitude estimation error $\tilde{R}$.
		In fact, the convergence of $\tilde{r}_i$ to zero implies the convergence of $r_i-\tilde{R}\hat{r}_i$ to zero. This result, together with the convergence of $\tilde{R}$ to $I_3$, allows to conclude that $\hat{r}_i\to r_i$ as $t\to \infty$.  
		Moreover, the global (exponential) convergence of $\tilde{r}$ plays an important role in the stability analysis of the attitude estimation error, which can be seen as a vanishing disturbance in the error dynamics of the attitude estimation in \eqref{eqn:aiss}.
	\end{rem}

	Let us define the extended state space $\mathcal{S}_o:=SO(3) \times  \mathbb{R}^{3N} \times   [0,T_M^1] \times \cdots \times [0,T_M^N]$ and the closed set $\mathcal{A}_o: = \{(\tilde{R},\tilde{r},\tau) \in \mathcal{S}_o: \tilde{R}=I_3,   \tilde{r}=0 \}$. Now, one can state the following main result:
	\begin{thm}\label{thm:hybrid_observer_I}
		Consider the hybrid closed-loop system \eqref{eqn:tau} and \eqref{eqn:hybrid_observer_I} with $k_o>0$ and $0<k_r<1$. Suppose that Assumption \ref{assum:non-collinear} and \ref{assum:intermittent} hold and the matrix $A$ defined in \eqref{eqn:def_A} is positive definite with three distinct eigenvalues. Then,  the set  $\mathcal{A}_o$ is almost globally asymptotically stable for the hybrid closed-loop system.
	\end{thm}
	\begin{proof}
		The proof of Theorem \ref{thm:hybrid_observer_I} can be conducted by using the almost global ISS property of the $\tilde{R}$-subsystem from Proposition \ref{prop:AISS-SO3} and the global exponential stability of the $(\tilde{x},\tau)$-subsystem from Proposition \ref{prop:tilde_r}, as well as the results of \cite[Lemma 2]{wang2021cdc} (adapted from \cite[Theorem 2]{angeli2004almost}).  
		From \eqref{eqn:dot_tildeR},  one can rewrite the dynamics of the estimation error $\tilde{R}$  in the form of \eqref{eqn:aiss} with $\bar{\Gamma}(\hat{R}) = -k_o
        \begin{bmatrix}
				\rho_1 r_1^\times \hat{R}\T,   \rho_2   r_2^\times \hat{R}\T, \dots, \rho_N   r_N^\times \hat{R}\T
		\end{bmatrix} $ and $u = [\tilde{r}_1\T R,\tilde{r}_2\T R,\dots,\tilde{r}_N\T R]\T$. Then, according to Proposition \ref{prop:AISS-SO3}, the equilibrium $I_3$ is almost globally input-to-state stable with respect to the input $u$ using the facts that $\|\bar{\Gamma}(X)\|_F^2 
		= \sum_{i=1}^N 2\rho_i^2 \|r_i\|^2$ for all $X \in SO(3)$ and  $A$ being positive definite with three distinct eigenvalues. Additionally, according to Proposition \ref{prop:tilde_r} and applying the fact $\|u\|^2= \|\tilde{r}|^2$, it follows that the input $u$ converges globally exponentially to zero from all initial conditions. Therefore, as per \cite[Lemma 2]{wang2021cdc}, the overall closed-loop system is endowed with AGAS guarantees, \ie, the estimated states $(\hat{R},\hat{r_i})$  converge asymptotically to the actual ones $(R,r_i)$ from almost all initial conditions except from a set of zero Lebesgue measure.
	\end{proof}
	
	\begin{rem}
        It's worth pointing out that the design of our proposed observer \eqref{eqn:hybrid_observer0} is independent of the knowledge of the upper bound $T_M^i$. As per Assumption 2, we only assume the existence (not the knowledge) of a finite upper bound $T_M^i$ (not necessarily small) guaranteeing that the time interval between two consecutive vector measurements is finite. 
	\end{rem}
	
	\begin{rem}
		Due to the above-mentioned topological obstruction to GAS on $SO(3)$, the best stability result one can achieve for the proposed observer \eqref{eqn:hybrid_observer0} is AGAS. This motivates us to redesign the hybrid estimation scheme leading to robust and global stability results as shown in the following section.
	\end{rem}
	
	\section{Hybrid Observer with GAS Guarantees} \label{sec:hybrod_obs_GAS}
	
	\subsection{Observer Design}
	It is well-known that global asymptotic stability results are precluded on $SO(3)$ with time-invariant smooth vector fields. This is due to the existence of additional undesired equilibrium points other than the identity. To achieve global attitude estimation on $SO(3)$, in terms of intermittent vector measurements, we propose a new hybrid estimation scheme relying on a switching mechanism proposed in \cite{wang2021hybrid}. We first introduce an auxiliary switching variable $\theta \in \mathbb{R}$ whose dynamics consist of continuous flows and discrete jumps. More specially, it flows to drive the attitude estimation error $\tilde{R}$ toward the equilibrium points of its closed-loop system, and jumps to some nonzero value in a nonempty finite set $\Theta\subset \mathbb{R}$, leading to the minimum value of a cost function, when the attitude estimation error is in the neighbourhood of the undesired equilibrium points. These carefully designed hybrid dynamics of the switching variable $\theta$ (hybrid auxiliary system), together with an appropriate innovation term for the attitude estimation, guarantee that the attitude estimation error will converge (globally asymptotically) to the desired equilibrium point $I_3$ (more details about this hybrid strategy can be found in \cite{wang2021hybrid}). Now, we propose the following hybrid observer:
	\begin{subequations}\label{eqn:hybrid_observer1}
		\begin{align}
			& \dot{\hat{R}} = \hat{R}(\omega + k_o\hat{R}\T \sigma_R)^\times \label{eqn:hybrid_observer1-1}                                                                     \\
			& \begin{cases}
				\dot{\hat{r}}_i ~ = k_o \sigma_R^\times \hat{r}_i,              &  t \neq t_k^i,~    k\in \mathbb{N} \\
				\hat{r}_i^+   =  \hat{r}_i +    k_{r} (\hat{R}{b}_i-\hat{r}_i), &  t = t_k^i,~  k\in \mathbb{N}
			\end{cases}
			\label{eqn:hybrid_observer1-2}     \\
			& \begin{cases}
				\dot{\theta} ~~ =- k_\theta\left( \gamma \theta + 2u\T \mathcal{R}_u\T (\theta) \sigma_R \right), &  (\theta, \hat{r})\in \mathcal{F}_\theta \\
				\theta^+   \in  \{\bar{\theta}\in \Theta: \mu_\phi(\bar{\theta},\hat{r})=0 \},                    &  (\theta, \hat{r})\in \mathcal{J}_\theta
			\end{cases}  \label{eqn:hybrid_dynamics_theta}
		\end{align}
	\end{subequations}
	where $k_o, k_\theta, \gamma >0, 0<k_r<1$, $\Theta \subset \mathbb{R}$ denotes a nonempty finite set, and
	\begin{subequations}\label{eqn:hybrid_dynamics_design}
		\begin{align}
			\mathcal{F}_\theta       & :=  \{(\theta,\hat{r})\in   \mathbb{R}\times \mathbb{R}^{3N}  : \mu_\phi(\theta,\hat{r}) \leq \delta \}                                                              \\
			\mathcal{J}_\theta       & :=  \{(\theta,\hat{r})\in    \mathbb{R} \times \mathbb{R}^{3N}: \mu_\phi(\theta,\hat{r}) \geq \delta \}                                                              \\
			\mu_\phi(\theta,\hat{r}) & := \phi(\theta,\hat{r}) - \min_{\theta'\in \Theta} \phi(\theta',\hat{r})                                                                                    \\
			\phi(\theta,\hat{r})     & := \frac{1}{2}  \sum\nolimits_{i=1}^N \rho_i \|r_i - \mathcal{R}_u\T(\theta)\hat{r}_i\|^2 + \frac{\gamma}{2} \theta^2 \label{eqn:hybrid_dynamics_design-4}
		\end{align}
	\end{subequations}
	with some constant $\delta>0$ and $\hat{r}:=[\hat{r}_1\T,\hat{r}_2\T,\cdots,\hat{r}_N\T]\T \in \mathbb{R}^{3N}$. The set of parameters $\mathcal{P}_A:=\{\Theta,k_\theta,\gamma,u,\delta\}$ for the design of the hybrid dynamics \eqref{eqn:hybrid_observer1}-\eqref{eqn:hybrid_dynamics_design} will be given later. The innovation term $\sigma_R$ is designed as
	\begin{align}
		\sigma_R & =   \sum_{i=1}^N \rho_i     \hat{r}_i  \times \mathcal{R}_u (\theta)  r_i \label{eqn:def_sigma_R}
	\end{align}
	where $\rho_i>0$ for all $i\in \mathbb{I}$ and $\mathcal{R}_u(\theta) := \mathcal{R}_a(u,\theta) = \exp(\theta u^\times), u\in \mathcal{S}^2$.
	The structure of our proposed hybrid observer \eqref{eqn:hybrid_observer1} is given in Fig. \ref{fig:hybrid_observer}.
	\begin{figure}[!ht]
		\centering
		\tikzset{every picture/.style={line width=0.75pt}} 
		
		\begin{tikzpicture}[x=0.43pt,y=0.48pt,yscale=-1,xscale=1] \small
			
			\draw   (100,126) -- (242,126) -- (242,171.27) -- (100,171.27) -- cycle ;
			\draw   (330,129.27) -- (493,129.27) -- (493,175.27) -- (330,175.27) -- cycle ;
			\draw   (329,197.27) -- (493.4,197.27) -- (493.4,243.27) -- (329,243.27) -- cycle ;
			\draw    (243,148.36) -- (328,148.28) ;
			\draw [shift={(330,148.28)}, rotate = 539.94] [color={rgb, 255:red, 0; green, 0; blue, 0 }  ][line width=0.75]    (10.93,-4.9) .. controls (6.95,-2.3) and (3.31,-0.67) .. (0,0) .. controls (3.31,0.67) and (6.95,2.3) .. (10.93,4.9)   ;
			\draw    (280.6,148.36) -- (281,219.28) ;
			\draw    (281,219.28) -- (327,219.24) ;
			\draw [shift={(329,219.24)}, rotate = 539.96] [color={rgb, 255:red, 0; green, 0; blue, 0 }  ][line width=0.75]    (10.93,-4.9) .. controls (6.95,-2.3) and (3.31,-0.67) .. (0,0) .. controls (3.31,0.67) and (6.95,2.3) .. (10.93,4.9)   ;
			\draw    (492.4,152.15) -- (559.75,152.27) ;
			\draw [shift={(561.75,152.27)}, rotate = 180.1] [color={rgb, 255:red, 0; green, 0; blue, 0 }  ][line width=0.75]    (10.93,-4.9) .. controls (6.95,-2.3) and (3.31,-0.67) .. (0,0) .. controls (3.31,0.67) and (6.95,2.3) .. (10.93,4.9)   ;
			\draw    (493.9,218.65) -- (516.6,218.6) ;
			\draw    (516.6,218.6) -- (517,261.28) ;
			\draw    (178,261.28) -- (517,261.28) ;
			\draw    (178,261.28) -- (178.2,173.4) ;
			\draw [shift={(178.2,171.4)}, rotate = 450.13] [color={rgb, 255:red, 0; green, 0; blue, 0 }  ][line width=0.75]    (10.93,-4.9) .. controls (6.95,-2.3) and (3.31,-0.67) .. (0,0) .. controls (3.31,0.67) and (6.95,2.3) .. (10.93,4.9)   ;
			\draw    (532,151.77) -- (532.2,277) ;
			\draw    (138.6,276.2) -- (532.2,277) ;
			\draw    (138.6,276.2) -- (138.6,173.8) ;
			\draw [shift={(138.6,171.8)}, rotate = 450] [color={rgb, 255:red, 0; green, 0; blue, 0 }  ][line width=0.75]    (10.93,-4.9) .. controls (6.95,-2.3) and (3.31,-0.67) .. (0,0) .. controls (3.31,0.67) and (6.95,2.3) .. (10.93,4.9)   ;
			\draw  [dash pattern={on 4.5pt off 4.5pt}]  (14,145.6) -- (100,145.79) ;
			\draw [shift={(100.2,145.8)}, rotate = 180.17] [color={rgb, 255:red, 0; green, 0; blue, 0 }  ][line width=0.75]    (10.93,-4.9) .. controls (6.95,-2.3) and (3.31,-0.67) .. (0,0) .. controls (3.31,0.67) and (6.95,2.3) .. (10.93,4.9)   ;
			\draw    (413.8,198.04) -- (413.53,177.75) ;
			\draw [shift={(413.5,175.75)}, rotate = 449.23] [color={rgb, 255:red, 0; green, 0; blue, 0 }  ][line width=0.75]    (10.93,-4.9) .. controls (6.95,-2.3) and (3.31,-0.67) .. (0,0) .. controls (3.31,0.67) and (6.95,2.3) .. (10.93,4.9)   ;
			\draw    (413,105) -- (413,125.86) ;
			\draw [shift={(413,127.86)}, rotate = 271.79] [color={rgb, 255:red, 0; green, 0; blue, 0 }  ][line width=0.75]    (10.93,-4.9) .. controls (6.95,-2.3) and (3.31,-0.67) .. (0,0) .. controls (3.31,0.67) and (6.95,2.3) .. (10.93,4.9)   ;

			\draw (102,129) node [anchor=north west][inner sep=0.75pt]   [align=left] {Hybrid vector\\ estimation \eqref{eqn:hybrid_observer1-2}};
			\draw (332,132.27) node [anchor=north west][inner sep=0.75pt]   [align=left] {Attitude\\ estimation  \eqref{eqn:hybrid_observer1-1}};
			\draw (331,200.27) node [anchor=north west][inner sep=0.75pt]   [align=left] {Hybrid auxiliary\\ system \eqref{eqn:hybrid_dynamics_theta}};
			\draw (14,123) node [anchor=north west][inner sep=0.75pt]   [align=left] {$b_1,...,b_N$};
			\draw (243,126) node [anchor=north west][inner sep=0.75pt]   [align=left] {$\hat{r}_1,...,\hat{r}_N$};
			\draw (498,128) node [anchor=north west][inner sep=0.75pt]   [align=left] {$\hat{R}$};
			\draw (498,199) node [anchor=north west][inner sep=0.75pt]   [align=left] {$\theta$};
			\draw (420,178) node [anchor=north west][inner sep=0.75pt]   [align=left] {$\theta$};
			\draw (420,110) node [anchor=north west][inner sep=0.75pt]   [align=left] {$\omega$};
		\end{tikzpicture}
		\caption{The architecture of the proposed hybrid observer \eqref{eqn:hybrid_observer1}. 
		}\label{fig:hybrid_observer}
	\end{figure}
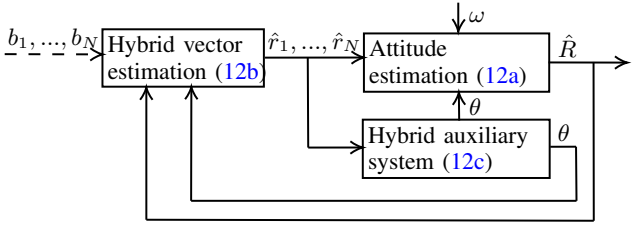

	The main difference between the innovation term $\sigma_R$ designed in \eqref{eqn:def_sigma_R} with respect to the one designed in \eqref{eqn:def_sigma_R_0} is the additional term $\mathcal{R}_u(\theta)$, which is specially designed such that $\sigma_R$ will be reset to some nonzero values after the jump of $\theta$ when $\theta=0$ and the attitude estimation error is close to one of the undesired equilibrium points of the closed-loop system. This term  $\sigma_R$, together with the hybrid dynamics of $\theta$, are instrumental in achieving GAS with our proposed hybrid observer.
		
	Note that the functions $\phi$ and $\mu_\phi$ in \eqref{eqn:hybrid_dynamics_design} are designed to determine the jump of the switching variable $\theta$ using the flow and jump sets defined in \eqref{eqn:hybrid_dynamics_design}. The function $\phi$ defined in  \eqref{eqn:hybrid_dynamics_design-4} can be seen as a cost function as $\hat{r}_i$  converges exponentially to $\hat{R}b_i$. One can verify that  $\phi(\theta,\hat{r}) = 0$ if the estimation error and the switching variable $\theta$ satisfy $\tilde{R}=I_3, \theta=0,\tilde{r}=0$.
	Moreover, the jump map \eqref{eqn:hybrid_dynamics_theta} is designed to ensure that the cost function $\phi$ has a strict decrease after each jump, and the flow map of \eqref{eqn:hybrid_dynamics_theta} motivated from \cite{wang2021hybrid} is designed such that the estimation errors converge to the equilibrium points of the overall closed-loop system in the flow set.

	The set of parameters $\mathcal{P}_A$ is designed as follows:
	\begin{equation}\label{eqn:setP}
		\mathcal{P}_{A}: \begin{cases}
			\Theta   & = \{|\theta_i|\in (0,\pi], i = 1,\dots,m\in \mathbb{N}\}     \\
			k_\theta & > 0       \\
			u        & =  \alpha_1 v_1^A + \alpha_2 v_2^A + \alpha_3 v_3^A \in \mathbb{S}^2         \\
			\gamma   & <\frac{4\Delta^*}{\pi^2}      \\
			\delta   & < (\frac{4\Delta^*}{\pi^2} - \gamma)\frac{\theta_M^2}{2}, \theta_M:=\max_{\theta'\in \Theta} |\theta'|
		\end{cases}
	\end{equation}
	where scalars $\alpha_1,\alpha_2,\alpha_3 \in [0,1]$ with $\sum_{i=1}^3 \alpha_i^2 =1$ and $\Delta^*>0$ are given as per one of the following cases:
	\begin{enumerate}
		\item If $\lambda_1^A=\lambda_2^A$, $\alpha_3^2= 1-\frac{\lambda_2^A}{\lambda_3^A}$ and $\Delta^* = \lambda_1^A(1-\frac{\lambda_2^A}{\lambda_3^A})$;
		\item If $\lambda_1^A < \frac{\lambda_1^A \lambda_3^A}{\lambda^A_3 - \lambda_1^A} \leq  \lambda_2^A $,
		$
		\alpha_i^2 = \frac{\lambda_i^A}{\lambda_2^A +\lambda_3^A}$ for all $i\in \{2,3\}$ and $\Delta^* = \lambda_1^A$;
		\item  If  $ \lambda_1^A <  \lambda_2^A < \frac{\lambda_1^A \lambda_3^A}{\lambda^A_3 - \lambda_1^A}$,
		$
		\alpha_i^2 = 1-  \frac{4}{\varSigma_A}\prod_{j\neq i} \lambda_j^A$ for all $ i\in \{1,2,3\}$ and $
		\Delta^* =  \frac{4}{\varSigma_A} \prod_{j} \lambda_j^A$ with $\varSigma_A = \sum_{l=1}^3\sum_{k\neq l}^3\lambda_l^A\lambda^A_k$.
	\end{enumerate}
	with $(\lambda^A_i , v^A_i)$ denoting the $i$-th pair of eigenvalue-eigenvector of the matrix $A$ and $0<\lambda^A_1\leq \lambda^A_2 < \lambda^A_3$.
	
	\begin{rem}
		The design of $\mathcal{P}_{A}$ is inspired by \cite[Proposition 2]{berkane2017construction} and adapted from  \cite[Proposition 3]{wang2021hybrid}. 
		Note that the choice of the unit vector $u$ in \eqref{eqn:setP}, relying on the eigenvalues and eigenvectors of the matrix $A$, is optimal in terms of $\Delta^*:=\max_{u\in \mathbb{S}^2} (\min_{v\in \mathcal{E}(A)}\Delta(v,u))$  with $\Delta(u,v) = u\T \left(\tr(A)I_3 - A - 2v\T A v(I_3 - vv\T) \right) u$. Note also that small $\gamma$, together with large $\Delta^*$, results in a large gap $\delta$ for the design of the flow and jump sets in \eqref{eqn:hybrid_dynamics_design} (strengthening the robustness to measurement noise). However, small $\gamma$ may slow down the convergence of the auxiliary state $\theta$ as per the flow of \eqref{eqn:hybrid_dynamics_theta}, leading to lower convergence rates for the overall system (see simulation example in \cite{wang2021hybrid}). Hence, in practical applications, the parameter $\gamma$ should be chosen as a suitable trade-off between robustness and speed of convergence, \ie, the measurement noise level and the convergence rate of the overall system.
	\end{rem}

	\subsection{Stability Analysis}
	For the sake of simplicity, we introduce the functions $\tilde{\phi}:SO(3) \times \mathbb{R} \times \mathbb{R}^{3N}  \to \mathbb{R}_{\geq 0}$, $\mu_{\tilde{\phi}} :SO(3) \times \mathbb{R} \times \mathbb{R}^{3N}  \to \mathbb{R}$ and $ {\sigma}  :SO(3) \times \mathbb{R} \times \mathbb{R}^{3N}  \to \mathbb{R}^3$ given as
	\begin{subequations} \label{eqn:def_tilde_phi_mu_phi}
		\begin{align}
			\tilde{\phi}(X,\theta,v)       & := \frac{1}{2}  \sum_{i=1}^N \rho_i \|(I_3- X\mathcal{R}_u(\theta))\T r_i \nonumber                            \\
			& \qquad   \qquad + (X\mathcal{R}_u(\theta))\T v_i\|^2   + \frac{\gamma}{2} \theta^2                             \\
			\mu_{\tilde{\phi}}(X,\theta,v) & := \tilde{\phi}(X,\theta,V) - \min_{\theta'\in \Theta} \tilde{\phi}(X,\theta',v)                               \\
			{\sigma}(X,\theta,v)           & := \mathcal{R}_u(\theta)\psi(AX\mathcal{R}_u(\theta)) +  \Gamma(X,\theta) v \label{eqn:def_tilde_phi_mu_phi-3}
		\end{align}
	\end{subequations}
	where $(X,\theta)\in SO(3) \times \mathbb{R}$, $v = [v_1\T,v_2\T,\dots,v_N\T]\T \in \mathbb{R}^{3N}$ and $\Gamma(X,\theta)=[\rho_1    (\mathcal{R}_u(\theta) r_1)^\times   X\T,   \rho_2   (\mathcal{R}_u(\theta) r_2)^\times X\T ,   \cdots, \\ \rho_N   (\mathcal{R}_u(\theta) r_N)^\times X\T ]$.
	Substituting $\hat{r}_i = \tilde{R}\T(r_i - \tilde{r}_i)$ in \eqref{eqn:hybrid_dynamics_design} and \eqref{eqn:def_sigma_R}, the functions $\phi$  and $\mu_\phi$ defined in \eqref{eqn:hybrid_dynamics_design} can be rewritten in terms of the estimation errors as
	\begin{align}
		\mu_\phi (\theta,\hat{r}) & =\mu_{\tilde{\phi}}(\tilde{R},\theta,\tilde{r}) \\  \phi(\theta,\hat{r}) &= \tilde{\phi}(\tilde{R},\theta,\tilde{r})
	\end{align}
	and the innovation term $\sigma_R$ in \eqref{eqn:def_sigma_R} can be rewritten as
	\begin{align}
		\sigma_R
		& = -  \sum_{i=1}^N \rho_i  (\mathcal{R}_u  (\theta) r_i)^\times \hat{r}_i \nonumber                                                                                                               \\
		& = -\mathcal{R}_u  (\theta) \sum_{i=1}^N \rho_i  r_i^\times   \mathcal{R}_u\T(\theta) \tilde{R}\T (r_i-\tilde{r}_i)   \nonumber                                                                   \\
		& = -\mathcal{R}_u(\theta)  \sum_{i=1}^N \rho_i \left(  r_i^\times   (\tilde{R}\mathcal{R}_u(\theta))\T r_i   - r_i^\times \mathcal{R}_u\T(\theta)   \tilde{R}\T  \tilde{r}_i \right)    \nonumber \\
		& = \mathcal{R}_u(\theta)\psi(A\tilde{R}\mathcal{R}_u(\theta)) +  \Gamma(\tilde{R},\theta) \tilde{r} \nonumber                                                                                     \\
		& = \sigma (\tilde{R},\theta,\tilde{r})   \label{eqn:error_sigma_R}
	\end{align}
	where $A$ is defined in \eqref{eqn:def_A} and $\sigma$ is defined in \eqref{eqn:def_tilde_phi_mu_phi-3}, and we made use of the facts that $x \times y  = x^\times y = -y^\times x$, $(Rx)^\times   = R x^\times R\T  $ and  $\sum_{i=1}^N \rho_i   r_i^\times   R\T r_i = -\psi(AR)  $   for all $R\in SO(3), x,y\in \mathbb{R}^3$.

	From \eqref{eqn:dynamics_R}, \eqref{eqn:hybrid_observer1-1}, \eqref{eqn:hybrid_dynamics_theta} and \eqref{eqn:error_sigma_R}, the hybrid dynamics of the state $x_o := (\tilde{R},\theta) \in SO(3)\times \mathbb{R}:=\mho$   are given as
	\begin{align} \label{eqn:hybrid_observer_error}
		\begin{cases}
			\dot{x}_o~ = F_o(x_o,\tilde{r}), &  (x_o,\tilde{r})\in \mathcal{F}_o \\
			x_o^+ \in G_o(x_o,\tilde{r}),    &  (x_o,\tilde{r})\in \mathcal{J}_o
		\end{cases}
	\end{align}
	where the flow and jump sets $\mathcal{F}_o,\mathcal{J}_o$, and the flow and jump maps $F_o, G_o$ are defined as
	\begin{subequations}\label{eqn:def_Fo_Jo}
		\begin{align}
			\mathcal{F}_o      & := \{(x_o,\tilde{r})\in \mho \times  \mathbb{R}^{3N} : \mu_{\tilde{\phi}}(\tilde{R},\theta,\tilde{r}) \leq \delta \} \label{eqn:def_Fo_Jo-1}  \\
			\mathcal{J}_o      & :=  \{(x_o,\tilde{r})\in \mho \times  \mathbb{R}^{3N} : \mu_{\tilde{\phi}}(\tilde{R},\theta,\tilde{r}) \geq \delta \} \label{eqn:def_Fo_Jo-2} \\
			F_o(x_o,\tilde{r}) & := \begin{pmatrix}
				\tilde{R} (-k_o\sigma (\tilde{R},\theta,\tilde{r})  )^\times \\
				- k_\theta ( \gamma \theta + 2 u\T \mathcal{R}_u\T (\theta) \sigma (\tilde{R},\theta,\tilde{r})   )
			\end{pmatrix}                             \\ 
			G_o(x_o,\tilde{r}) & :=
			\begin{pmatrix}
				\tilde{R} \\
                    \{{\theta}'\in \Theta: \mu_{\tilde{\phi}}(\tilde{R},\theta',\tilde{r})=0  \}
			\end{pmatrix}.
		\end{align}
	\end{subequations}

	\begin{lem}	\label{lem:Psi}
		Consider the definitions of $\tilde{\phi}, \mu_{\tilde{\phi}}$ in \eqref{eqn:def_tilde_phi_mu_phi} and the set of parameters   $\mathcal{P}_A$ designed in \eqref{eqn:setP}. Then, the following inequality holds:
		\begin{equation}
			\mu_{\tilde{\phi}}(\tilde{R},\theta,\tilde{r}) > \delta, \quad  \forall (\tilde{R},\theta,\tilde{r}) \in \Psi_A \label{eqn:mu_delta}
		\end{equation}
		where the set $\Psi_A \subset   \mho \times  \mathbb{R}^{3N}$ is defined as
		\begin{align}
			\Psi_A & :=\{(\tilde{R},\theta,\tilde{r})\in \mho \times \mathbb{R}^{3N}: \nonumber                                       \\
			& ~~\qquad \tilde{R}  = \mathcal{R}_a(\pi,v), v\in \mathcal{E}(A),  \theta=0, \tilde{r}=0 \}.  \label{eqn:def_Psi}
		\end{align}
	\end{lem}
	
	The proof of Lemma \ref{lem:Psi} can be easily conducted from \cite[Proposition 2]{wang2021hybrid}  using the facts that $\mu_{\tilde{\phi}}(\tilde{R},\theta,\tilde{r}) = \tilde{\phi}(\tilde{R},0,0) - \min_{\theta'\in \Theta} \tilde{\phi}(\tilde{R},\theta',0) 
	= \tr((I_3-\tilde{R})A) - \min_{\theta'\in \Theta}   ( \tr((I_3-\tilde{R}\mathcal{R}_u(\theta'))A )  + \frac{\gamma}{2} \theta'^2  ) > \delta $ for each $(\tilde{R},\theta,\tilde{r})\in \Psi_A $. This result, together with the definitions of the flow and jump sets $\mathcal{F}_o, \mathcal{J}_o$ in \eqref{eqn:def_Fo_Jo-1} and \eqref{eqn:def_Fo_Jo-2}, implies that all the points in the set $\Psi_A$ are only located in the jump set $\mathcal{J}_o$ (\ie, $\Psi_A \cap \mathcal{F}_o = \emptyset$ and $\Psi_A \subset \mathcal{J}_o$).

	\begin{prop}\label{prop:x_o}
		Consider the hybrid system \eqref{eqn:hybrid_observer_error} with $\tilde{r}(t)\equiv 0$ for all $t\geq 0$. Suppose that Assumption \ref{assum:non-collinear} holds, and choose the set $\mathcal{P}_A$ as in \eqref{eqn:setP},  $k_o>0$ and $0<k_r<1$. Then,  the equilibrium point $(I_3,0)$ is globally asymptotically stable for the hybrid system \eqref{eqn:hybrid_observer_error}.
	\end{prop}
	\begin{proof}
		See Appendix \ref{sec:x_o}.
	\end{proof}
	\begin{rem}
		The key of the proof is to show that the Lyapunov function is non-increasing in the flow set and strictly decreasing in the jump set, and all the undesired equilibria are located in the jump set. Hence, the equilibrium point $(I_3,0)$ is globally asymptotically stable for the hybrid  system \eqref{eqn:hybrid_observer_error} with $\tilde{r}(t)\equiv 0$ for all $t\geq 0$. Note that $\tilde{r}$  converges (globally exponentially) to zero as shown in Proposition \ref{prop:tilde_r}.
	\end{rem}

	Let us define the extended state space $\mathcal{S}:=SO(3)\times  \mathbb{R} \times  \mathbb{R}^{3N} \times   [0,T_M^1] \times \cdots \times [0,T_M^N]$ and the closed set $\mathcal{A}: = \{(\tilde{R},\theta,\tilde{r},\tau) \in \mathcal{S}: \tilde{R}=I_3, \theta=0,\tilde{r}=0 \}$.
	Now, one can state the following main result:
	\begin{thm}\label{thm:hybrid_observer}
		Consider the hybrid closed-loop system \eqref{eqn:tau}, \eqref{eqn:hybrid_tilde_r} and \eqref{eqn:hybrid_observer_error}. 	Suppose that Assumption \ref{assum:non-collinear}-\ref{assum:intermittent} hold, and the matrix $A$ defined in \eqref{eqn:def_A} is positive definite with $0<\lambda^A_1\leq \lambda^A_2 < \lambda^A_3$. Choose the set of parameters   $\mathcal{P}_A$ as in \eqref{eqn:setP},  $k_o>0$ and $0<k_r<1$. Then,  the set  $\mathcal{A}$ is globally asymptotically stable for the hybrid closed-loop system.
	\end{thm}
	\begin{proof}
		See Appendix \ref{sec:hybrid_observer}.
	\end{proof}
	\begin{rem}
		The key of the proof of Theorem \ref{thm:hybrid_observer} is to show that the set  $\mathcal{A}$ is stable and globally attractive for the hybrid closed-loop system using the results of \cite{sanfelice2007invariance,goebel2012hybrid}. Unlike the (hybrid) observers developed on vector space $\mathbb{R}^n$ in \cite{ferrante2016state, sferlazza2018time, ferrante2021observer, landicheff2022continuous, chen2022variable}, the proposed hybrid observer on manifold $SO(3)\times \mathbb{R}^n$ consists of two types of hybrid strategies to obtain global and smooth attitude estimation. The first one involves a continuous-discrete estimation scheme to tackle the intermittent vector measurements. The second one introduces a switching mechanism to overcome the topological obstruction to GAS on $SO(3)$. This makes the proof of our Theorem \ref{thm:hybrid_observer} more challenging with respect to the existing hybrid observers such as \cite{wu2015globally,berkane2017construction,berkane2017hybrid,berkane2019attitude,wang2021hybrid}, where only one type of jump is involved. It is worth pointing out that the proposed hybrid observer designed on the manifold $SO(3)\times \mathbb{R}^n$ can be extended to different systems evolving on manifolds, for instance, pose (attitude and position) estimation on $SE(3)$, pose and linear velocity estimation on $SE_2(3)$, and SLAM (pose, linear velocity and landmark position estimation) on $SE_{n+2}(3)$. 
	\end{rem}

	\section{Simulation and Experimental Results} \label{sec:simulation_experiments}
	\subsection{Simulation Results}
	In this subsection, some simulation results are presented to illustrate the performance of the proposed hybrid observers. For comparison purposes, we consider the proposed hybrid observer \eqref{eqn:hybrid_observer0} (referred to as `Proposed'), the predict-update hybrid observer proposed in \cite{berkane2019attitude} (referred to as `Hybrid \cite{berkane2019attitude}'), and the nonlinear complementary filter of \cite{mahony2008nonlinear} with the zero-order-hold method (referred to as `CF \cite{mahony2008nonlinear} + ZOH') given as
	\begin{subequations}
		\begin{align}
			\dot{\hat{R}} & = \hat{R}(\omega + k_P\hat{R}\T \sigma_R)^\times  \\
			\sigma_R      & =   \sum_{i=1}^N k_i \hat{R} b_i^m  \times   r_i  \\
			b_i^m(t)      & = b_i(t_k^i), \quad  t\in [t_k^i, t_{k+1}^i)
		\end{align}
	\end{subequations}
	where $k_i>0, k_P>0$. For the sake of simplicity, we consider the angle $\vartheta= \frac{180}{\pi} \arccos(\frac{1}2(\tr(\tilde{R})-1))$ as a representation of the attitude estimation errors plotted in the figures. 	
	
	We assume that there are three inertial vectors, chosen as $r_1 = [\sqrt{2}/2~ \sqrt{2}~ 0]\T, r_2 = [\sqrt{2}/2~ -\sqrt{2}/2~ 0]\T$ and $r_3 = [0~ 0~ -1]\T$, available for measurement, whose (noisy) measurements are written as $b_i = R\T r_i + n_b$ with $n_b$ denoting the zero mean white  Gaussian noise and $\text{Cov}(n_b) = \sigma I_3$.  The sampling rate of the vector measurements $b_1, b_2$ and $b_3$ are around $f_1\approx 10$Hz ($T_m^1 = 0.09, T_M^1 = 0.11$), $f_2\approx 20$Hz ($T_m^2 = 0.04, T_M^2 = 0.06$) and $f_3\approx 50$Hz ($T_m^3 = 0.01, T_M^3 = 0.03$), respectively. The attitude $R(t)$ of the rigid body is integrated using $\omega(t) = \omega_o[\sin(0.1t)~ \sin(0.1t + \pi/3)~  \cos(0.5t)]\T$ (rad/s) with $\omega_o=2$ denoting the amplitude and the initial attitude $R(0) = I_3$. We assume that the angular velocity $\omega$ is quasi-continuous and the differential equations for each observer are solved using the fourth-order Runge-Kutta method with a fixed step size of 1 millisecond (\ie, running at 1000Hz).
	The same initial attitude is considered for each observer with $\hat{R}(0) = \mathcal{R}_a(\pi/2, v), v = [0.8~0.6~0]\T$, and $\hat{r}_i(0) = r_i, \theta(0)=0$. The gain parameters are carefully tuned such that all the observers have a similar convergence rate (see Fig. \ref{fig:figure1}), where $\rho_1 = 0.2, \rho_2 = 0.3, \rho_3 = 0.5, k_r = 0.45, k_o = 15$ for the proposed hybrid observer, $\rho_1 = 0.03, \rho_2 = 0.0450, \rho_3 = 0.075$ for the `Hybrid  \cite{berkane2019attitude}' and $k_i=\rho_i, k_P = 12$ for the `CF \cite{mahony2008nonlinear} + ZOH'.

	\begin{figure}[!ht]
		\centering
		~~~~\includegraphics[width=0.95\linewidth]{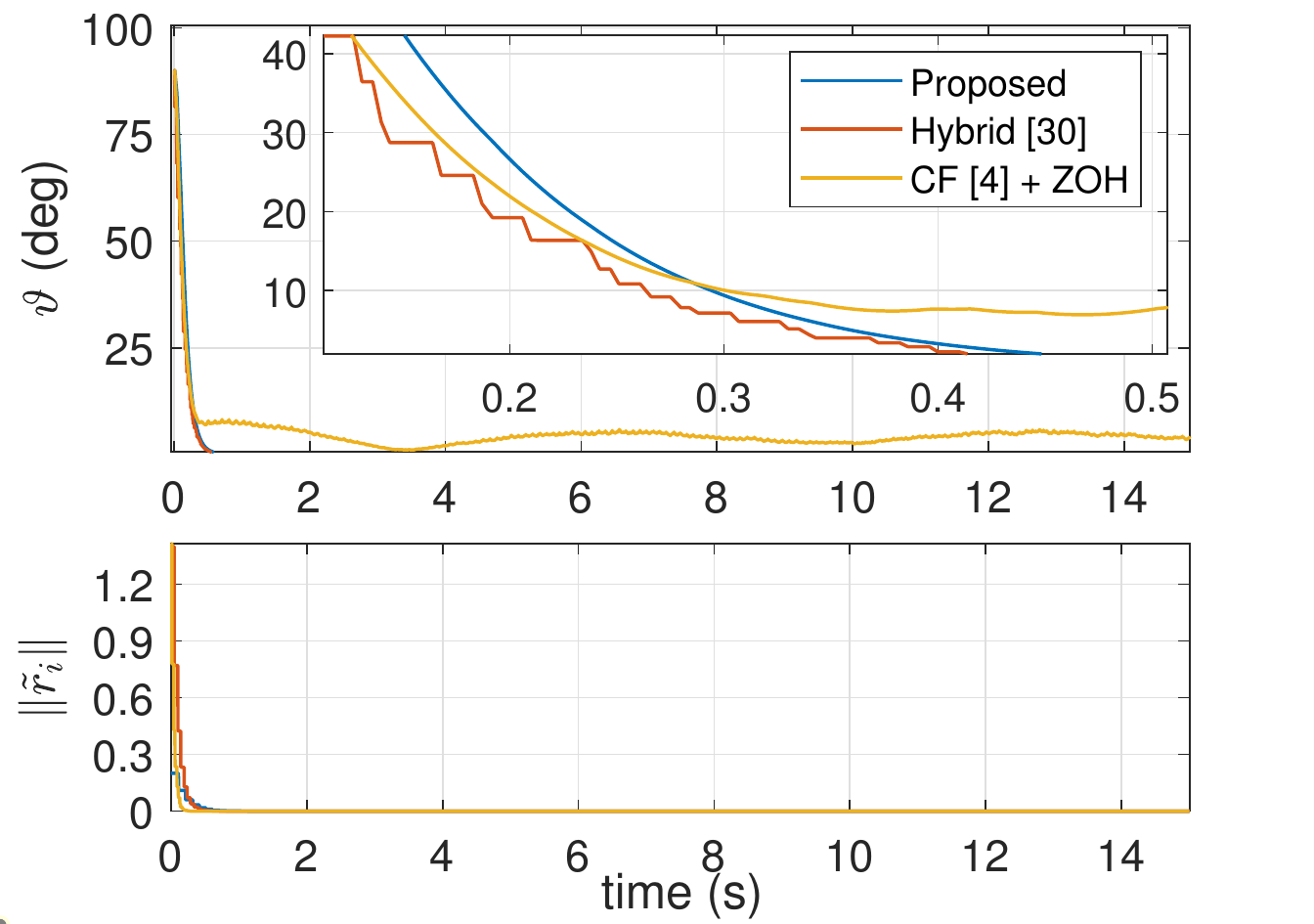}
		\caption{Time evolution of the attitude and vector estimation errors for different attitude observers with noise-free measurements ($\sigma=0$).}
		\label{fig:figure1}
	\end{figure}

	\begin{figure}[!ht]
		\centering
		\includegraphics[width=0.95\linewidth]{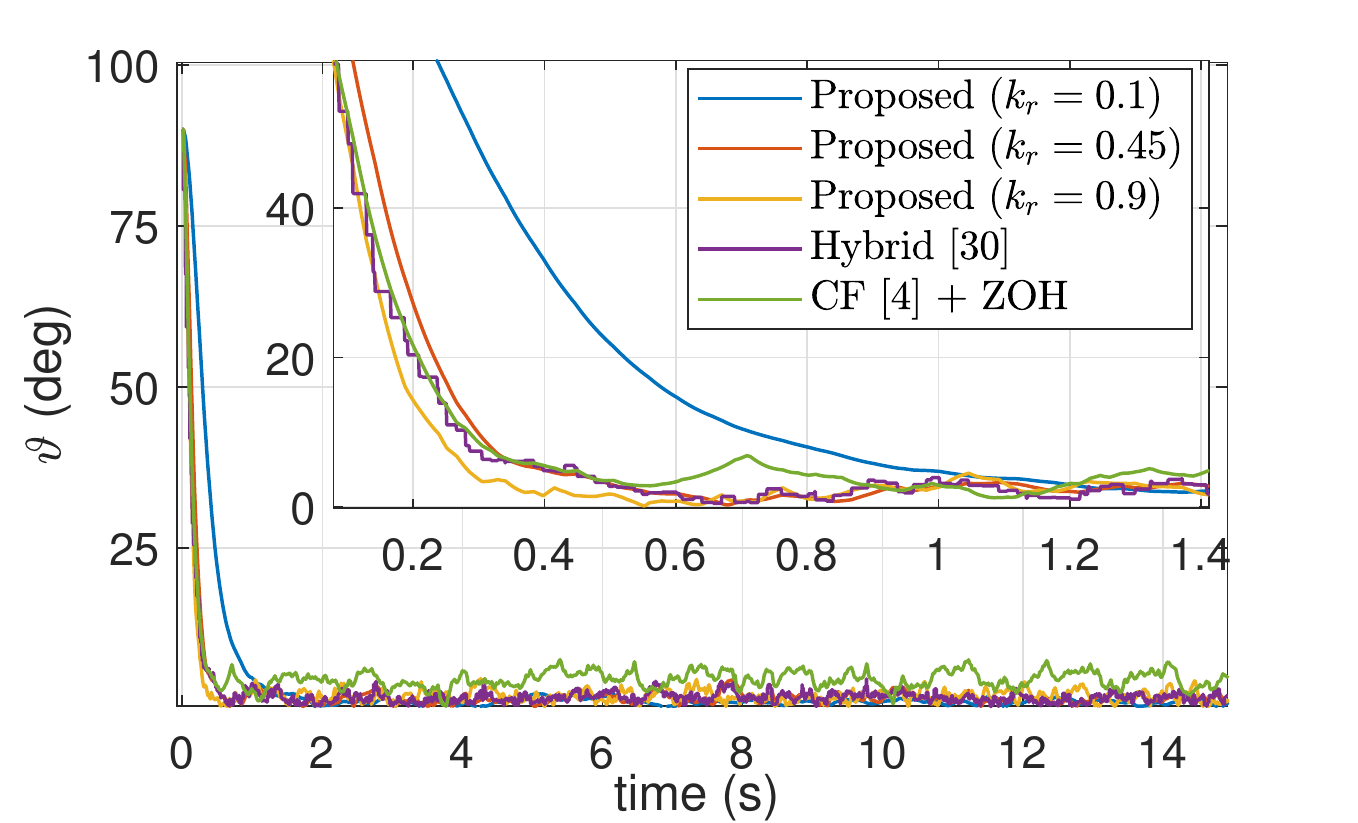}
		\caption{Time evolution of the attitude estimation errors for the proposed attitude observer with different gain parameter $k_r$ in the presence of measurements noise ($\sigma=0.08$).}
		\label{fig:figure2}
	\end{figure}
	
	\begin{table*}[!ht]
		\centering 
		\caption{Averaged Attitude Estimation Error}\label{tab:1}
		\setlength\tabcolsep{5pt} 
		\begin{tabular}{l|c|c|c}   %
			\hline
			& \textbf{Proposed}  & \textbf{Hybrid \cite{berkane2019attitude}} & \textbf{CF \cite{mahony2008nonlinear} + ZOH} \\
			\hline
			\hline
			Test 1 ($\sigma\,{=}\,0, \omega_o\,{=}\,2$)                        & \textbf{  0 deg}   & \textbf{ 0 deg}                            & 4.36 deg                                     \\
			Test 2 ($\sigma\,{=}\,0, \omega_o\,{=}\,5$)                        & \textbf{ 0 deg}    & \textbf{ 0 deg}                            & 11.35 deg                                    \\
			Test 3 ($\sigma\,{=}\,0.08, \omega_o\,{=}\,2$)                     & \textbf{1.92 deg}  & 1.95 deg                                   & 4.40 deg                                     \\
			Test 4 ($\sigma\,{=}\,0.08, \omega_o\,{=}\,5$)                     & \textbf{1.67  deg} & 1.82 deg                                   & 11.28  deg                                   \\
			Test 5 ($\sigma\,{=}\,0, \omega_o\,{=}\,2, f_2\,{\approx}\,10$)    & \textbf{0  deg}    & \textbf{0 deg}                             & 5.92  deg                                    \\
			Test 6 ($\sigma\,{=}\,0.08, \omega_o\,{=}\,2, f_2\,{\approx}\,10$) & {1.91  deg}        & \textbf{1.88 deg}                          & 5.96  deg                                    \\
			\hline
		\end{tabular}
	\end{table*}
	
	\begin{figure*}[!ht]
		\centering
		\includegraphics[width=0.95\linewidth]{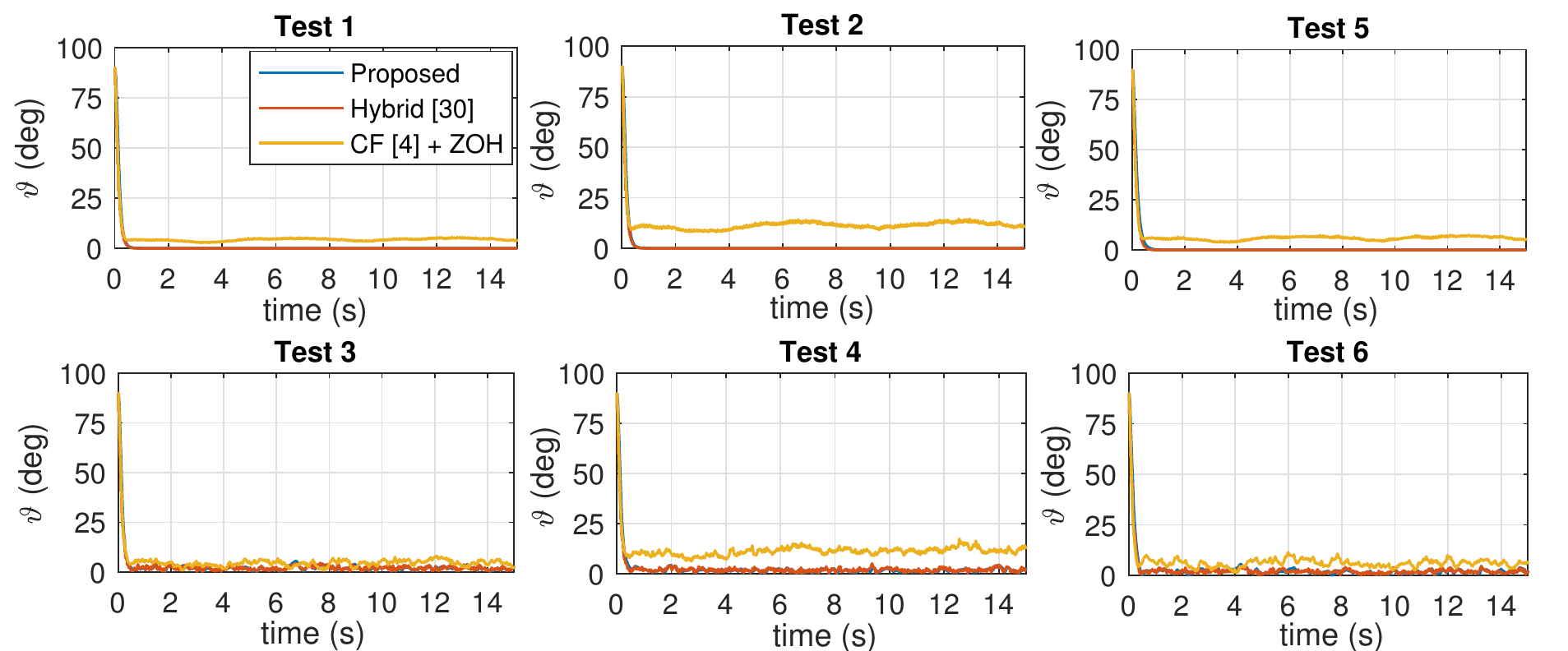}
		\caption{Time evolution of the attitude estimation errors for the attitude observers with different values of the noise covariance, angular velocity amplitude and sampling rate.}
		\label{fig:figure3}
	\end{figure*}

	In the first simulation, as we can see from Fig. \ref{fig:figure1}, all the observers have the same convergence rate in the case noise-free case. Moreover, both our proposed hybrid observer and the `Hybrid  \cite{berkane2019attitude}' guarantee zero estimation error as $t\to \infty$. This indicates that the convergence of the continuous observer, implemented with intermittent measurements, is not guaranteed. In the second simulation, we implemented our hybrid observer with different values of the gain parameter $k_r$ and vector measurements noise with $\sigma=0.08$. From Fig. \ref{fig:figure2}, we can see that large values of $k_r$ lead to fast convergence rates and large estimation errors. In fact, the hybrid dynamics of $\hat{r}_i$ in \eqref{eqn:hybrid_observer0-2} play the role of a low pass filter for the vector measurements and the gain parameter $k_r$ can be tuned based on the measurement noise and the converge speed requirements. In the third simulation, 6 different tests are presented with different values of the noise covariance, angular velocity amplitude and vector measurements sampling rates. The time evolution of the attitude estimation errors for different attitude observers is shown in Fig. \ref{fig:figure3}. Moreover, the averaged attitude estimation errors (after 2 seconds) for each observer are given in Table \ref{tab:1}. As one can see, our proposed hybrid observer and the `Hybrid \cite{berkane2019attitude}' have zero estimation error in the tests without measurement noise, and the static error of the `CF \cite{mahony2008nonlinear} + ZOH' increases when the motion is fast and the sampling rates are low (see Test 1, 2 and 5). We can also notice that our proposed observer provides the best results in most of the tests in the presence of measurement noise except for Test 6.
	
	\subsection{Experiments Results}%
	Our observer has been tested experimentally using Intel RealSense D435i which consists of an RGB-D camera (providing color images and depth images) and an IMU (including a gyroscope and an accelerometer). In our experimental setup, the RGB-D camera provides color and depth images at a frequency of 30 FPS (frames per second) with a resolution of $640 \times 480$ and the gyroscope data output rate is 400 Hz.

	In this test, we made use of the AprilTag markers \cite{wang2016apriltag} and considered the tag corners as landmarks in the environment. In this setup, a single AprilTag marker can provide four (non-aligned) landmarks. Let $p_i^{\mathcal{I}}, i=1,2,3,4$ denote the inertial-frame positions of the four landmarks generated from the corners of an AprilTag marker. Then, five inertial vectors are constructed as follows:
    \begin{align}
    	  r_i = \frac{p_i^{\mathcal{I}}-p_c^{\mathcal{I}}}{\|p_i^{\mathcal{I}}-p_c^{\mathcal{I}}\|},\ \forall i=1,2,3,4, \quad 
    	  r_5  = \frac{r_1 \times r_2}{\|r_1 \times r_2\|}
    \end{align}
	where $p_c^{\mathcal{I}} = \frac{1}{4}\sum p_i^{\mathcal{I}}$ is the center of the tag marker, and the additional vector $r_5$ is orthogonal to the plane of the tag. 
	For the sake of simplicity, we consider the inertial frame $\{\mathcal{I}\}$ attached to the center of the AprilTag marker, and the body frame $\{\mathcal{B}\}$ attached to the camera D435i and aligned with the IMU sensor (see Fig. \ref{fig:diagram_1}). 
	Hence, from the inertial frame setup, one obtains the following five inertial vectors: $
	    r_1 = [
	        {-\sqrt{2}}/{2}~
	        {-\sqrt{2}}/{2}~
	        0]^\top,
	    r_2 = [
	        {-\sqrt{2}}/{2}~
	        {\sqrt{2}}/{2}~
	        0]^\top,
	    r_3 = [
	        {\sqrt{2}}/{2}~
	        {\sqrt{2}}/{2}~
	        0]^\top,  
	    r_4 =[
	        {\sqrt{2}}/{2}~
	        {-\sqrt{2}}/{2}~
	        0]^\top$, $r_5 =[
	        0~
	        0~
	        -1]^\top$. 
	Let $(u_i,v_i)$ and $d_i$ denote the pixel and depth measurements of the $i$-th landmark, respectively (see Fig. \ref{fig:diagram_2}). Then, according to the structure of the RGB-D camera D435i, the body-frame positions of the corners are given by
	\begin{align} \label{eqn:p_i_B}
        p_i^{\mathcal{B}} = d_i \begin{bmatrix}
        (u_i-c_x)/f_x \\
        (v_i-c_y)/f_y \\
        1
        \end{bmatrix},\quad \forall i=1,2,3,4, 
	\end{align}
	where the focal length $(f_x,f_y)$ and the optical center $(c_x,c_y)$ are the internal parameters obtained from camera calibration. 
	Moreover, the corresponding body-frame vector measurements are obtained as follows: 
    \begin{align}\label{eqn:b_i_measure}
    	  b_i = \frac{p_i^{\mathcal{B}}-p_c^{\mathcal{B}}}{\|p_i^{\mathcal{B}}-p_c^{\mathcal{B}}\|}, \  \forall i=1,2,3,4, \quad  
    	    b_5  = \frac{b_1 \times b_2}{\|b_1 \times b_2\|}
    \end{align} 
	with $p_c^{\mathcal{B}} = \frac{1}{4}\sum p_i^{\mathcal{B}}$. The construction of the body-frame vector measurements has been summarized in 
	Algorithm \ref{algo:1}. Moreover, a detailed diagram of our experimental setup is given in Fig. \ref{fig:diagram_3}. 
	
	\begin{algorithm}   \small
	\caption{\small The construction of the body-frame vector measurements from an RGB-D camera}
	\begin{algorithmic}[1]  \label{algo:1}
        \STATE Read the color image;
        \STATE Read the depth image;
        \STATE Align the depth image to the color image;
        \STATE Detect the AprilTag in the image;
        \IF{AprilTag is detected}
            \STATE Obtain the pixels $(u_i,v_i)$ and depths $d_i$ of the tag corners;
            \STATE Compute the 3D positions $p_i^{\mathcal{B}}$ of landmarks using \eqref{eqn:p_i_B};
            \STATE Construct the body-frame vectors using \eqref{eqn:b_i_measure};
        \ENDIF 
    \end{algorithmic}
    \end{algorithm} 
	
	\begin{figure}
		\centering
		\begin{subfigure}{.46\linewidth}
			\centering
			\includegraphics[width=0.95\linewidth]{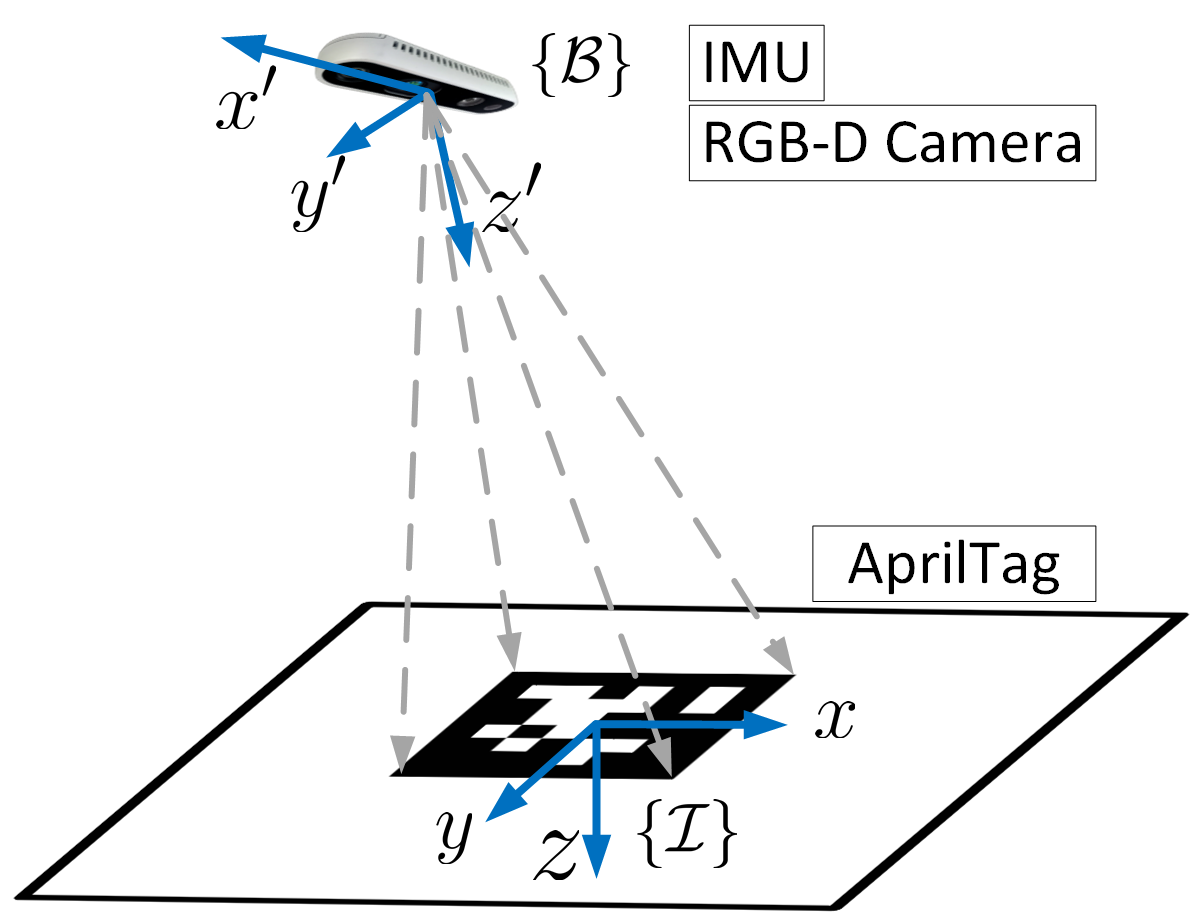}
			\caption{}\label{fig:diagram_1}
		\end{subfigure}
		\begin{subfigure}{.46\linewidth}
			\centering
			\includegraphics[width=0.95\linewidth]{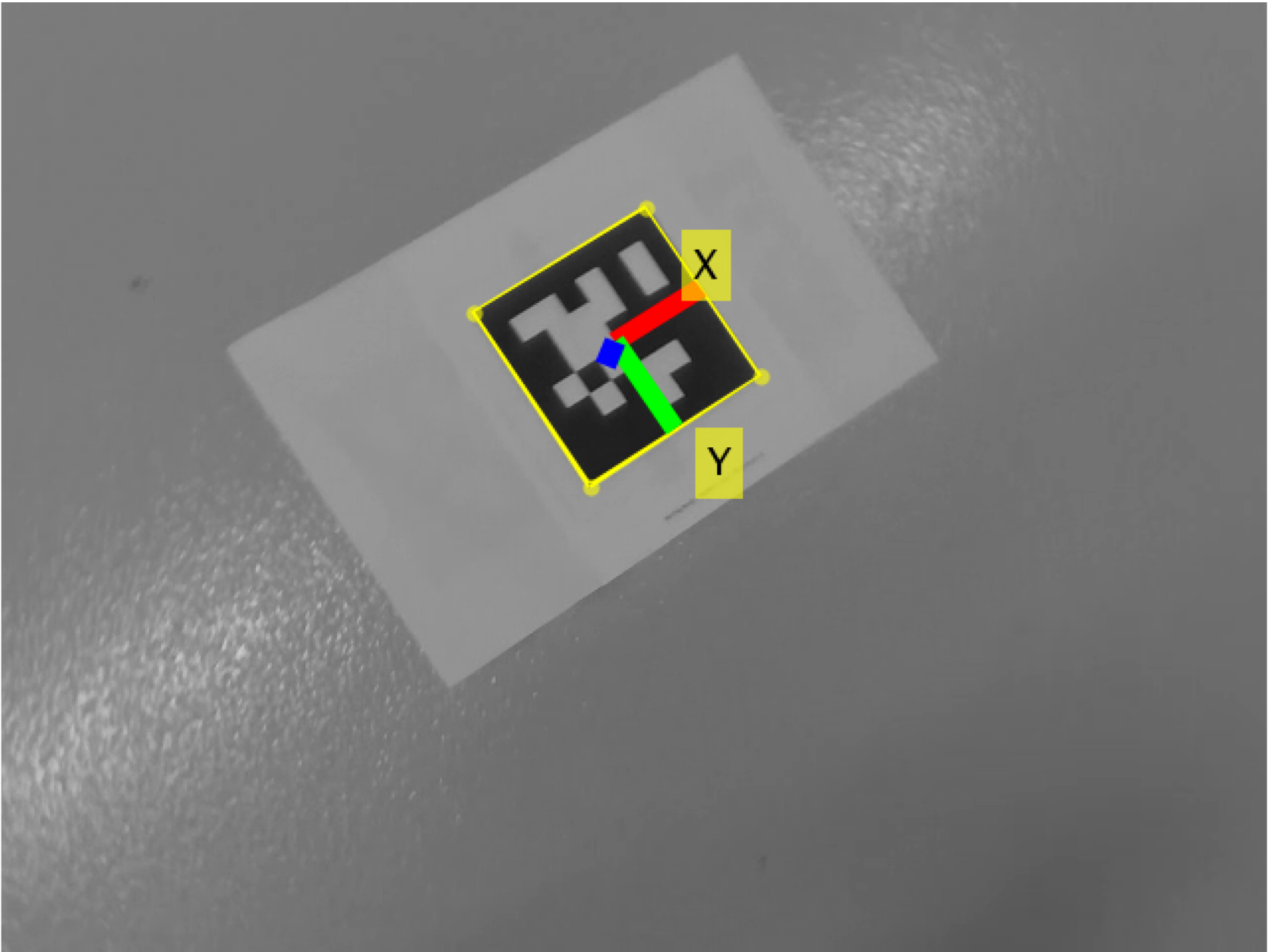}
			\caption{}\label{fig:diagram_2}
		\end{subfigure}
		\\[0.3em] 
		\begin{subfigure}{.94\linewidth}
			\centering
			\includegraphics[width=0.95\linewidth]{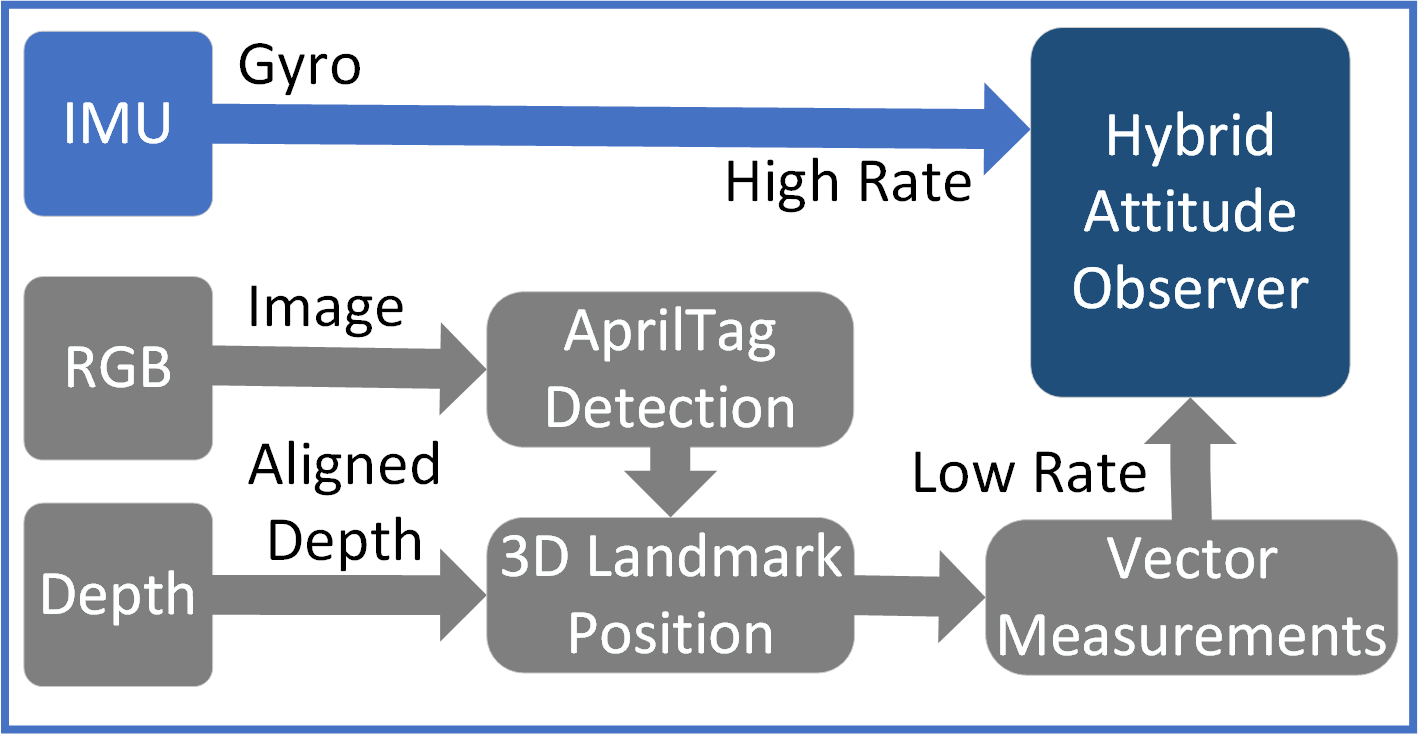}
			\caption{}\label{fig:diagram_3}
		\end{subfigure}
		\caption{The experimental setup and the diagram of the real-time implementation of the proposed hybrid attitude observer using an IMU sensor and an RGB-D camera as well as the AprilTag markers. A video of the experimental test is available at \href{https://youtu.be/9eBqG6p7CXk}{\tt   https://youtu.be/9eBqG6p7CXk}.}
		\label{fig:diagram}
	\end{figure}
	
	The same initial attitude is considered for each observer with $\hat{R}(0) = I_3$ and $\hat{r}_i(0) = r_i$ for all $i=1,2,\dots,5$. The gain parameters are tuned such that all the observers have the similar performance with $\rho_i = (6-i)/15, k_r = 0.5, k_o = 10.5$ for the proposed observer, $\rho_i = 2(6-i)/25$ for the `Hybrid  \cite{berkane2019attitude}' and $k_i= \rho_i, k_P = 9.5$ for the `CF \cite{mahony2008nonlinear} + ZOH'. We also consider a vision-only-based pose estimation method (referred to as `Vision only') using the command {\tt readAprilTag} provided by the Computer Vision Toolbox, which provides relative pose estimation from a single AprilTag of known scale. 
	
	In this experimental test, the data from the RealSense D435i are transmitted to a PC using a USB 3.0 cable. The proposed attitude observer is implemented online using the Simulink Coder (generates and executes C/C++ code from a Simulink model) on an Intel Core i7-3540M running at 3.0 GHz, while the other methods are implemented offline using the collected same data (including gyro, landmark positions and images).  
	Since the ground truth is not available in this experimental setup, alternatively, we consider the following root mean square error (RMSE):
	\begin{align}
		RMSE = \sqrt{\frac{1}{N}\sum\nolimits_{i=1}^N \|\hat{R}\T r_i - b_i\|^2}.
	\end{align} 
	The experimental results are presented in Fig. \ref{fig:realtimeattitude}. It can be observed that the three attitude observers have a similar performance in this experimental test, which is better than the `Vision only' method in terms of noise, due to the nature of filtering by taking into account the information of gyro. The RMSE of each attitude observer starting from a large vale converges, after a few seconds, to the vicinity of zero. Moreover, the zoomed plot in Fig. \ref{fig:realtimeattitude} also shows that the proposed hybrid observer provides smoother attitude estimation compared to the `Hybrid  \cite{berkane2019attitude}' in the presence of measurement noise. This is a benefit from the fact that the proposed observer generates continuous estimates of the attitude.

	\begin{figure}[!ht]
		\centering
		\includegraphics[width=0.95\linewidth]{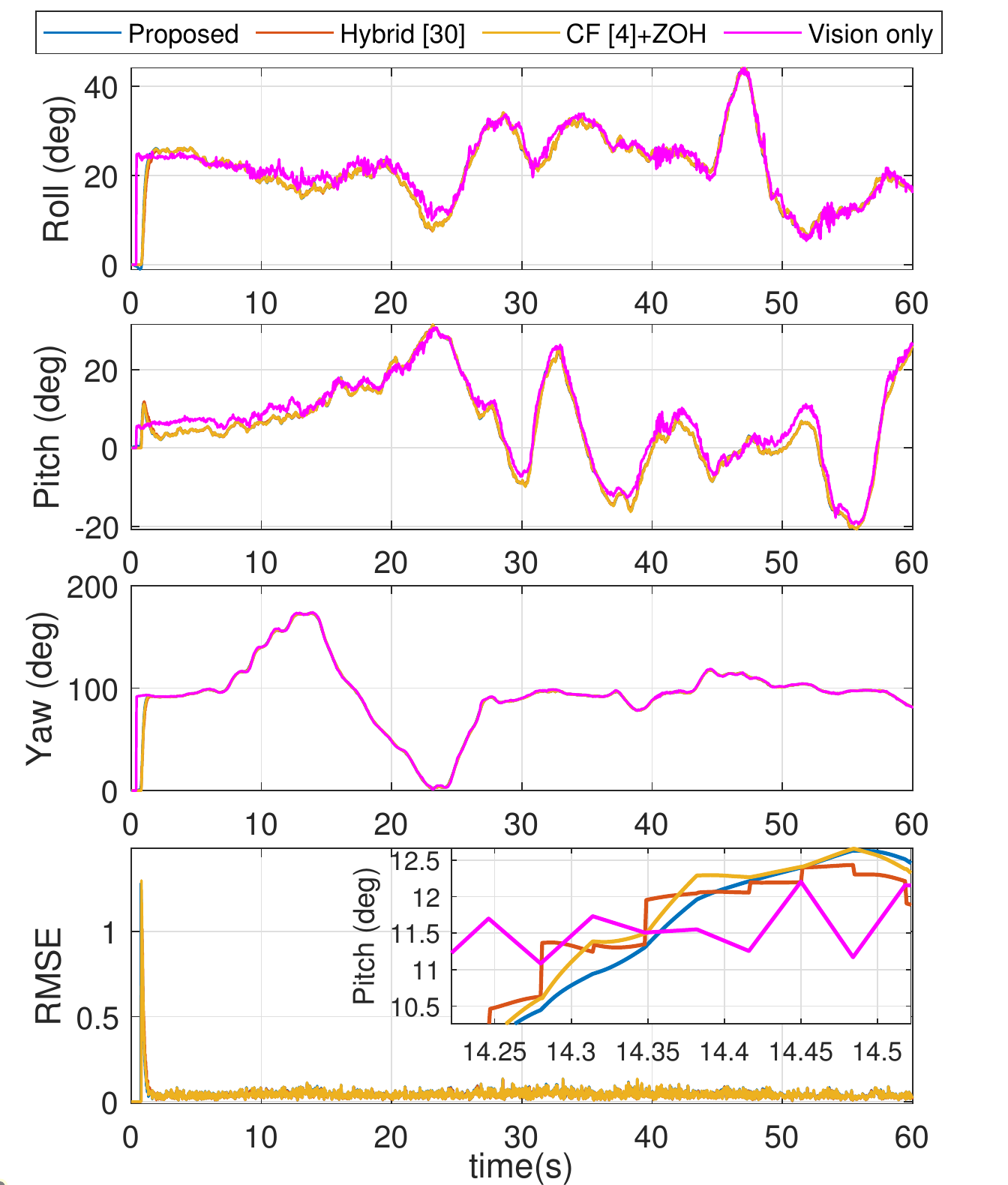}
		\caption{Time evolution of the attitude estimation and estimation errors using the high-rate gyro measurements (400 Hz) and the low-rate RGB-D camera measurements (30 FPS).}
		\label{fig:realtimeattitude}
	\end{figure}

	\section{Conclusion}
	In this work, we addressed the problem of nonlinear attitude estimation with intermittent and multi-rate vector measurements by designing two types of hybrid nonlinear attitude observers.  
	The first hybrid attitude observer, endowed with almost global asymptotic stability, provides continuous estimates of the attitude through the use of a kinematic model on $SO(3)$ driven by continuous gyro measurements and a hybrid input generated from intermittent and multi-rate body-frame vector measurements. 
	The second hybrid attitude observer provides stronger stability results by handling the undesired equilibria through the introduction of an additional switching mechanism motivated from \cite{wang2021hybrid} leading to global asymptotic stability guarantees. To the best of our knowledge, this is the first work dealing with nonlinear smooth attitude estimation, with GAS guarantees, involving intermittent and multi-rate vector measurements. 
	A set of numerical simulation tests show that the proposed hybrid observer exhibits better performance than the continuous nonlinear complementary filter \cite{mahony2008nonlinear} with the ZOH method and the (non-smooth) predict-update hybrid observer of \cite{berkane2019attitude}. To illustrate the real-time performance of the proposed approach, we implemented the proposed observer online using high-rate gyro measurements and low-rate vector measurements constructed from an RGB-D camera.

	\appendices

	\section{Proof of Proposition \ref{prop:tilde_r}} \label{sec:tilde_r}
	From \eqref{eqn:tau}-\eqref{eqn:hybrid_tilde_r}, for each vector $i\in \mathbb{I}$ one obtains the following hybrid system:
	\begin{align}
		\underbrace{\begin{pmatrix}
				\dot{\tilde{r}}_i \\
				\dot{\tau}_i
			\end{pmatrix} = \begin{pmatrix}
				0 \\
				-1
		\end{pmatrix}}_{ \tau_i\in [0,t_M^i]}~~
		\underbrace{\begin{pmatrix}
				{\tilde{r}}_i^+ \\
				{\tau}_i^+
			\end{pmatrix} \in  \begin{pmatrix}
				(1-k_r)\tilde{r}_i \\
				[T_m^i,T_M^i]
		\end{pmatrix}}_{ \tau_i\in \{0\}}  \label{eqn:hybrid_r}.
	\end{align}
	Note that the hybrid system \eqref{eqn:hybrid_r} is autonomous and satisfies the hybrid basic conditions of \cite[Assumption 6.5]{goebel2012hybrid}.
	For each vector $i\in \mathbb{I}$, consider the following Lyapunov function candidate:
	\begin{equation}
		V_r^{i}(\tilde{r}_i,\tau_i) =  e^{\mu \tau_i}   \tilde{r}_i\T \tilde{r}_i  \label{eqn:def_V1}
	\end{equation}
	where $0<\mu<  -\frac{2}{\bar{T}_M}\ln(1-k_r)   $ with $\bar{T}_M  := \max_{i\in \mathbb{I}} T_M^i$. Since $0<k_r<1$, one can easily verify that $-\ln(1-k_r)>0$. For all $\tau_i\in [0,T_M^i]$ (\ie, between two consecutive measurements of the $i$-th vector)  the time-derivative of $V_1$ along the flow of \eqref{eqn:hybrid_r} is given by
	\begin{align}
		\dot{V}_r^i(\tilde{r}_i,\tau_i) & =  - \mu   e^{\mu \tau_i} \tilde{r}_i\T \tilde{r}_i  =  -\mu V_r^i(\tilde{r}_i,\tau_i)   \label{eqn:dot_V1}.
	\end{align}
	When the measurement of the $i$-th vector arrives (\ie, $\tau_i=0$), one has $\tau_i^+\in [T_m^i,T_M^i], \tilde{r}_i^+   = (1-k_r) \tilde{r}_i$ and
	\begin{align}
		V_r^i(\tilde{r}_i^+,\tau_i^+)  
		& = e^{\mu \tau_i^+} (1-k_r)^2  \tilde{r}_i\T \tilde{r}_i  \nonumber             \\
		& \leq e^{\mu T_M^i} (1-k_r)^2  \tilde{r}_i\T \tilde{r}_i  \nonumber             \\
		& \leq e^{\mu \bar{T}_M} (1-k_r)^2   V_r^i(\tilde{r}_i,\tau_i) \label{eqn:V1+}
	\end{align}
	where we made use of $T_M^i \leq \bar{T}_M$ for all $i\in \mathbb{I}$.
	Let $\lambda_J: = e^{\mu \bar{T}_M} (1-k_r)^2$. Using the fact that  $\mu< -\frac{2}{\bar{T}_M}\ln(1-k_r) $, one can verify that $0<\lambda_J = e^{\mu \bar{T}_M } (1-k_r)^2    <1$.
	Hence, one can rewrite \eqref{eqn:V1+} as
	\begin{align}
		{V}_r^i(\tilde{r}_i^+,\tau_i^+) \leq  \lambda_J {V}_r^i(\tilde{r}_i,\tau_i) \leq e^{-\lambda} {V}_r^i(\tilde{r}_i,\tau_i)  \label{eqn:V1^+}
	\end{align}
	where $\lambda: = \min\{-\ln(\lambda_J), \mu\}>0$. Moreover, from \eqref{eqn:dot_V1} and \eqref{eqn:V1^+},  one can conclude that every maximal solution to the hybrid system \eqref{eqn:hybrid_r} is complete and satisfies $V_r^i(\tilde{r}_i(t,j),\tau_i(t,j)) \leq e^{- \lambda(t+j)} V_r^i(\tilde{r}_i(0,0),\tau_i(0,0))$ for all $(t,j) \in \dom (\tilde{r}_i,\tau_i)$. Let $\alpha = e^{\mu \bar{T}_M} $. Using the fact   $ \|\tilde{r}_i\|^2 \leq  V_r^i(\tilde{r}_i,\tau_i) \leq e^{\mu \bar{T}_M} \|\tilde{r}_i\|^2 $ from the definition of $V_r^i$ given in \eqref{eqn:def_V1}, one can further conclude that
	\begin{equation}
		\|\tilde{r}_i(t,j)\|^2 \leq \alpha  e^{- \lambda(t+j)}  \|\tilde{r}_i(0,0)\|^2  \nonumber
	\end{equation}
	which implies that for each $i\in \mathbb{I}$ the vector estimation error  $\tilde{r}_i$  converges (globally exponentially) to zero. This completes the proof.

	\section{Proof of Proposition \ref{prop:x_o}} \label{sec:x_o}
	Define $\mathcal{T}(\tilde{R},\theta): = \tilde{R}\mathcal{R}_u(\theta) \in SO(3)$. For simplicity, if no argument is indicated for $\mathcal{T}$, then it should be understood that $\mathcal{T}=\mathcal{T}(\tilde{R},\theta)$. For the sake of simplicity, let $\nu :=  \gamma  \theta  + 2   u\T\psi(A\mathcal{T})$, then one can show that
	\begin{align}
		\mathcal{R}_u\T (\theta) \sigma (\tilde{R},\theta,\tilde{r})
		& =  \psi(A\mathcal{T}) +    \overline{\Gamma}(\tilde{R},\theta)   \tilde{r}  \label{eqn:sigma_psi1} \\
		\gamma \theta + 2u\T \mathcal{R}_u\T (\theta) \sigma (\tilde{R},\theta,\tilde{r})
		& = \nu +  2   u\T \overline{\Gamma}(\tilde{R},\theta)   \tilde{r}  \label{eqn:sigma_psi2}
	\end{align}
	where $\overline{\Gamma}(\tilde{R},\theta) := \mathcal{R}_u\T (\theta)  \Gamma(\tilde{R},\theta)$. From \eqref{eqn:def_Fo_Jo}, \eqref{eqn:sigma_psi1} and \eqref{eqn:sigma_psi2}, the time-derivatives of $\theta$ and $\mathcal{T}(\tilde{R},\theta)$ are given by
	\begin{align}
		\dot{\theta} & = -k_\theta (\nu +  2   u\T \overline{\Gamma}(\tilde{R},\theta) \tilde{r})                                                                                  \\
		\dot{\mathcal{T}}
		& = 	\tilde{R}(  -k_o \sigma (\tilde{R},\theta,\tilde{r}) )^\times\mathcal{R}_u(\theta) +  \tilde{R} \mathcal{R}_u(\theta) (  \dot{\theta} u)^\times \nonumber \\
		& = 	\mathcal{T} ( - k_o \mathcal{R}_u\T (\theta) \sigma (\tilde{R},\theta,\tilde{r})  + \dot{\theta} u)^\times \nonumber                                      \\
		& = \mathcal{T} ( - k_o (\psi(A\mathcal{T}) + \overline{\Gamma}(\tilde{R},\theta)   \tilde{r} )  - k_\theta  u (\nu +  2   u\T\overline{\Gamma}(\tilde{R},\theta) \tilde{r}))^\times \nonumber                                                 \\
		& = \mathcal{T} ( - k_o \psi(A\mathcal{T}) - k_\theta \nu u -\widehat{\Gamma} (\tilde{R},\theta)   \tilde{r} )^\times
	\end{align}
	where $\widehat{\Gamma} (\tilde{R},\theta) := k_o\overline{\Gamma}(\tilde{R},\theta) + 2k_\theta uu\T \overline{\Gamma}(\tilde{R},\theta)$, 
    and we made use of the facts $\mathcal{R}_u(\theta)=e^{\theta u^\times}$, $\mathcal{R}_u(\theta)u = u$, $ \dot{\mathcal{R}}_u(\theta) =  \mathcal{R}_u(\theta) (\dot{\theta} u)^\times$ and $R\T a^\times R = (R\T a)^\times$ for all $R\in SO(3), a\in \mathbb{R}^3, u\in \mathbb{S}^2$. 
 
	Consider the following real-valued function $V_R: SO(3) \times \mathbb{R} \to \mathbb{R}$:
	\begin{align}
		V_R(x_o) 
		& =  \tr((I_3-\mathcal{T})A) + \frac{\gamma}{2} \theta^2    \label{eqn:defV_R}.
	\end{align}
	It is clear that $V_R(x_o)\geq 0$ for all $x_o\in SO(3)\times \mathbb{R}$ and $V_R(x_o)=0$ if and only if $x_o = (I_3,0)$. Thus, $V_R$ is positive definite with respect to the equilibrium point $(I_3,0)$. Then, the time-derivative of $V_R$ along the flows of \eqref{eqn:hybrid_observer_error}  is given as
	\begin{align}
		\dot{V}_R(x_o)
		& =  \tr(A\mathcal{T} (  k_o \psi(A\mathcal{T}) + k_\theta \nu u  + \widehat{\Gamma} (\tilde{R},\theta)   \tilde{r} )^\times )  \nonumber                            \\
		& \quad -  \gamma  k_\theta  \theta  (\nu +  2   u\T \overline{\Gamma}(\tilde{R},\theta) \tilde{r})   \nonumber                                                      \\
		& =  -2( k_o \psi(A\mathcal{T}) + k_\theta \nu u  + \widehat{\Gamma} (\tilde{R},\theta)   \tilde{r} )\T\psi(A\mathcal{T})   \nonumber                                \\
		& \quad   -  \gamma  k_\theta  \theta \nu    - 2 \gamma  k_\theta  \theta u\T \overline{\Gamma}(\tilde{R},\theta) \tilde{r})   \nonumber                             \\
		& =  -2k_o \|\psi(A\mathcal{T})\|^2 - 2 (\widehat{\Gamma} (\tilde{R},\theta)   \tilde{r} )\T\psi(A\mathcal{T})   \nonumber                                           \\
		& \quad    - k_\theta  \nu (\gamma   \theta + 2u\T \psi(A\mathcal{T}) )    - 2 \gamma  k_\theta  \theta u\T \overline{\Gamma}(\tilde{R},\theta) \tilde{r}  \nonumber \\
		& =  -2k_o \|\psi(A\mathcal{T})\|^2  - 2 (\widehat{\Gamma} (\tilde{R},\theta)   \tilde{r} )\T\psi(A\mathcal{T})   \nonumber                                          \\
		& \quad   - k_\theta  |\nu|^2   - 2 \gamma  k_\theta  \theta u\T \overline{\Gamma}(\tilde{R},\theta) \tilde{r}
	\end{align}
	where we have made use of the facts $\tr(AB)=\tr(BA)$, $\tr(A a^\times) = -2 a\T \psi(A)$  for all $A,B\in \mathbb{R}^{3\times 3}$ and $a\in \mathbb{R}^3$. From the definition of $\nu$, one has $\gamma \theta = \nu - 2u\T \psi(A\mathcal{T}) $. Moreover, from the definitions of $\Gamma(\tilde{R},\theta), \overline{\Gamma}(\tilde{R},\theta)$ and $\widehat{\Gamma} (\tilde{R},\theta) $, one can verify that there exist constants $\overline{c}_{\Gamma}, \widehat{c}_{\Gamma} >0 $ such that
	\begin{align*}
		\|\overline{\Gamma} (\tilde{R},\theta)   \tilde{r}\| \leq \overline{c}_{\Gamma} \|\tilde{r}\|,  \quad  \|\widehat{\Gamma} (\tilde{R},\theta)   \tilde{r}\| \leq \widehat{c}_{\Gamma} \|\tilde{r}\|
	\end{align*}
	for all $(\tilde{R},\theta,\tilde{r}) \in SO(3)\times \mathbb{R} \times \mathbb{R}^{3n}$. Then, the time-derivative of $V_R$ can be rewritten as
	\begin{align}
		\dot{V}_R(x_o)
		& =  -2k_o \|\psi(A\mathcal{T})\|^2  - 2 (\widehat{\Gamma} (\tilde{R},\theta)   \tilde{r} )\T\psi(A\mathcal{T})   \nonumber                 \\
		& \quad   - k_\theta  |\nu|^2   - 2  k_\theta (\nu - 2u\T \psi(A\mathcal{T}) ) u\T \overline{\Gamma}(\tilde{R},\theta) \tilde{r}  \nonumber \\
		& \leq  -2k_o \|\psi(A\mathcal{T})\|^2   - k_\theta  |\nu|^2   + 2  k_\theta \overline{c}_{\Gamma} |\nu| \|\tilde{r} \| \nonumber           \\
		& \quad    + 2(2 k_\theta \overline{c}_{\Gamma}  +  \widehat{c}_\Gamma )\|\psi(A\mathcal{T}) \|  \|\tilde{r} \|   \label{eqn:dotV_R}
	\end{align}
	where we made use of the facts $\|u\|=1$ for all $u\in \mathcal{S}^2$, $x\T y \leq \|x\| \|y\|$, $\|Rx\|=\|x\|$ for all $R\in SO(3), x, y\in \mathbb{R}^3$. Given $\tilde{r}\equiv 0$, one has $	\dot{V}_R(x_o)
	\leq  -2k_o \|\psi(A\mathcal{T})\|^2   - k_\theta  |\nu|^2$. Thus, from \eqref{eqn:dotV_R}, $V_R$ is non-increasing along the flows of \eqref{eqn:hybrid_observer_error} with $\tilde{r}\equiv 0$.
	
	On the other hand, from the definition of $\tilde{\phi}$ given in \eqref{eqn:def_tilde_phi_mu_phi}, one can show that
	\begin{align}
		\tilde{\phi}(x_o,\tilde{r}) -\sum_{i=1}^N \rho_i \|\tilde{r}_i\|^2 \leq V_R(x_o)	 \leq  \tilde{\phi}(x_o,\tilde{r})  +  \sum_{i=1}^N \rho_i \|\tilde{r}_i\|^2  \label{eqn:phi_V_R}
	\end{align}
	where we made use of the facts $ \|(I_3- \tilde{R}\mathcal{R}_u(\theta))\T r_i\|^2 -\|\tilde{r}_i\|^2 \leq   \|(I_3- \tilde{R}\mathcal{R}_u(\theta))\T r_i + (\tilde{R}\mathcal{R}_u(\theta))\T \tilde{r}_i\|^2 \leq \|(I_3- \tilde{R}\mathcal{R}_u(\theta))\T r_i\|^2 + \|\tilde{r}_i\|^2  $ and $\frac{1}{2}\sum_{i=1}^N \rho_i \|(I_3- \tilde{R}\mathcal{R}_u(\theta))\T r_i\|^2 + \frac{\gamma}{2} \theta^2  = \tr((I_3-\mathcal{T}(\tilde{R},\theta)A) + \frac{\gamma}{2} \theta^2 = V_R(x_o) $.
	Then, in view of \eqref{eqn:hybrid_observer_error}-\eqref{eqn:def_Fo_Jo}, \eqref{eqn:defV_R} and \eqref{eqn:phi_V_R}, for each jump in the jump set $\mathcal{J}_o$, one can show that
	\begin{align}
		& V_R(\tilde{R},\theta^+)  - V_R(\tilde{R},\theta) \nonumber                                                                                                                               \\
		& \leq \tilde{\phi}(\tilde{R},\theta^+,\tilde{r}) +  \sum_{i=1}^N \rho_i \|\tilde{r}_i\|^2 -\tilde{\phi}(\tilde{R},\theta,\tilde{r}) +  \sum_{i=1}^N \rho_i \|\tilde{r}_i\|^2 \nonumber    \\
		& =  -\left( \tilde{\phi}(\tilde{R},\theta,\tilde{r}) - \min_{\theta'\in \Theta} \tilde{\phi}(\tilde{R},\theta',\tilde{r})   \right)  +  2 \sum_{i=1}^N \rho_i \|\tilde{r}_i\|^2 \nonumber \\
		& \leq  - \delta   +  2 \sum_{i=1}^N \rho_i \|\tilde{r}_i\|^2  \label{eqn:V_R+}
	\end{align}
	where $  \theta^+ \in  	\{{\theta}'\in \Theta: \mu_{\tilde{\phi}}(\tilde{R},\theta',\tilde{r})=0  \}$ as per \eqref{eqn:def_Fo_Jo}.
	Thus, from \eqref{eqn:V_R+}, $V_R$ has a strict decrease over the jumps of \eqref{eqn:hybrid_observer_error} when $\tilde{r}\equiv 0$. Therefore, given $\tilde{r}\equiv 0$,  one concludes from \eqref{eqn:dotV_R} and \eqref{eqn:V_R+} that the equilibrium point $(I_3,0)$ is stable as per  \cite[Theorem 23]{goebel2009hybrid}, and any maximal solution to \eqref{eqn:hybrid_observer_error} is bounded. Moreover, it follows that $0\leq V_R(x_o(t,j)) \leq V_R(x_o(t_j,j))\leq V_R(x_o(t_j,j-1)) - \delta \leq V_R(x_o(0,0)) - j \delta$ with $(t,j)\succeq(t_j,j)\succeq (t_j,j-1)\succeq (0,0)$ and $(t,j), (t_j,j), (t_j,j-1)\in \dom x_o$. This implies that the number of jumps is finite (\ie, $\lceil \frac{V_R(0,0)}{\delta} \rceil$ with $\lceil \cdot\rceil$ denoting the ceiling function) and depends on the initial conditions.
	
	Next, we show that the equilibrium point $(I_3,0)$ is globally attractive. Applying the invariance principle for hybrid systems given in \cite[Theorem 4.7]{sanfelice2007invariance}, one concludes from \eqref{eqn:dotV_R} and \eqref{eqn:V_R+}  that any solution to the hybrid system \eqref{eqn:hybrid_observer_error} with $\tilde{r}\equiv 0$  must converge to the largest invariant set contained in $\mathcal{W}: = \{x_o=(\tilde{R},\theta)\in SO(3)\times \mathbb{R}|  \psi(A\mathcal{T}) = 0, \nu=0 \}$. Using the facts $\mathcal{T} = \tilde{R}\mathcal{R}_u(\theta)$ and $\nu = \gamma  \theta  + 2   u\T\psi(A\mathcal{T})$, for each $x_o\in \mathcal{W}$, one has $\theta=0$ and $\psi(A\tilde{R})=0$. From $\psi(A\tilde{R})=0$ with a symmetric positive definite matrix $A$ defined in \eqref{eqn:def_A}, one obtains $\tilde{R}\in \{I_3\}\cup \{\tilde{R} = \mathcal{R}_a(\pi, v), v\in \mathcal{E}(A)\}$. Thus, any solution $x_o$ to the hybrid system \eqref{eqn:hybrid_observer_error} with $\tilde{r}\equiv 0$ must converge to the largest invariant set contained in $\mathcal{W}': = \{(I_3,0)\}\cup\{x_o=(\tilde{R},\theta)\in SO(3)\times \mathbb{R}|  \tilde{R}=\mathcal{R}_a(\pi, v), v\in \mathcal{E}(A), \theta=0  \}$. From the definition of $\Psi_A$ defined in \eqref{eqn:def_Psi}, one has $\mathcal{W}'\times \{0\} = \{(I_3,0,0)\}\cup \Psi_A$. From the definitions of the sets $\mathcal{F}_o, \Psi_A$ and the fact $\mu_{\tilde{\phi}}(x_o,\tilde{r}) =   - \min_{\theta'\in \Theta} \tilde{\phi}(\tilde{R},\theta',\tilde{r}) \leq 0 $ as $(x_o,\tilde{r}) = (I_3,0,0)$, one obtains $\{(I_3,0,0)\} \subseteq  \mathcal{F}_o\cap \mathcal{W}'\times \{0\} $ and $\mathcal{F}_o\cap (\mathcal{W}'\times \{0\} \setminus \{(I_3,0,0)\}) = \mathcal{F}_o\cap \Psi_A = \emptyset$. Then, applying some simple set-theoretic arguments, one obtains $\mathcal{F}_o\cap \mathcal{W}'\times \{0\} \subseteq (\mathcal{F}_o\cap (\mathcal{W}'\times \{0\} \setminus \{(I_3,0,0)\})) \cup (\mathcal{F}_o\cap \{(I_3,0,0)\}) = \emptyset \cup   \{(I_3,0,0)\} = \{(I_3,0,0)\}$. Consequently, from the facts $\mathcal{F}_o\cap \mathcal{W}'\times \{0\} \subseteq  \{(I_3,0,0)\}$ and $\{(I_3,0,0)\} \subseteq  \mathcal{F}_o\cap \mathcal{W}'\times \{0\} $, one has $\{(I_3,0,0)\} =  \mathcal{F}_o\cap \mathcal{W}'\times \{0\}$, which implies that $\mathcal{W}'= \{(I_3,0)\}$ from the definition of $\mathcal{W}'$. Moreover, by virtue of \cite[Proposition 6.5]{goebel2012hybrid}, one can verify that every maximal solution to \eqref{eqn:hybrid_observer_error} with $\tilde{r}\equiv 0$ is complete. Hence, one can conclude that the equilibrium point $(I_3,0)$ is globally asymptotically stable for the hybrid system \eqref{eqn:hybrid_observer_error} with $\tilde{r}\equiv 0$. This completes the proof.

	\section{Proof of Theorem \ref{thm:hybrid_observer}}\label{sec:hybrid_observer}
	Define the extended state $x:= (x_o, \tilde{r},\tau) \in \mathcal{S}$.
	From \eqref{eqn:dynamics_R}, \eqref{eqn:tau}, \eqref{eqn:hybrid_observer1}, \eqref{eqn:hybrid_dynamics_design} and \eqref{eqn:error_sigma_R}, one obtains the following hybrid closed-loop system:
	\begin{equation}
		\mathcal{H}_o:
		\begin{cases}
			\dot{x}~~   = F(x),               &  x \in  {\mathcal{F}}                                 \\
			x^+ \,\in  \cup_{i=0}^N   G_i(x), &  x \in  {\mathcal{J}} :=\cup_{i=0}^N  {\mathcal{J}}_i
		\end{cases} 	\label{eqn:Hybrid-Closed-Loop}
	\end{equation}
	where $ {\mathcal{F}}:=\mathcal{F}_o  \times [0,T_M^1] \times \cdots \times [0,T_M^N]$ and 
	$ {\mathcal{J}}_0 = \mathcal{J}_o   \times [0, T_M^1] \times \cdots \times [0, T_M^N]$ and $ {\mathcal{J}}_i:= SO(3)\times \mathbb{R}  \times  \mathbb{R}^{3N} \times [0, T_M^1] \times \cdots  \times [0,T_M^{i-1}]\times \{0\} \times [0, T_M^{i+1}]\times \cdots \times [0,T_M^N]$ for all $i\in \mathbb{I}$, and the flow and jump maps are defined as follows:
	\begin{subequations}
		\begin{align}
			F(x)   & : =  (F_o(x_o,\tilde{r}),
			0_{3N\times 1},
			-1_{N\times 1}) \label{eqn:Hybrid-Closed-Loop-F} \\
			G_0(x) & : =  (G_o(x_o,\tilde{r}),
			\tilde{r},
			\tau )  \label{eqn:Hybrid-Closed-Loop-G-theta}   \\
			G_i(x) & : =  (x_o,
			G^{\tilde{r}}_i(\tilde{r}),
			G^{\tau}_i(\tau) )  \label{eqn:Hybrid-Closed-Loop-G-r}
		\end{align}
	\end{subequations}
	where $F_o(x_o,\tilde{r})$ and $G_o(x_o,\tilde{r})$ are defined in \eqref{eqn:def_Fo_Jo}, and
	\begin{align}
		G^{\tilde{r}}_i(\tilde{r}) & :=  \begin{medsize} \begin{bmatrix}
				\tilde{r}_1\T,\cdots, \tilde{r}_{i-1}\T, (1-k_r)\tilde{r}_i\T,\tilde{r}_{i+1}\T,\cdots,\tilde{r}_N\T
			\end{bmatrix}\T  \end{medsize} \\
		G^{\tau}_i(\tau)           & :=   \begin{medsize} \begin{bmatrix}
				\tau_1,\cdots, \tau_{i-1}, [T_m^i,T_M^i],\tau_{i+1},\cdots,\tau_N
			\end{bmatrix}\T \end{medsize}.
	\end{align}
	From \eqref{eqn:hybrid_tilde_r} and \eqref{eqn:hybrid_dynamics_theta}, one has $x^+\in G_i(x)$ if $x\in  {\mathcal{J}}_i$ for each $i\in \mathbb{I}_{\geq 0}: = \{0,1,2,\dots,N\}$. Note that in the hybrid dynamics \eqref{eqn:Hybrid-Closed-Loop}, $x$ could belong concurrently to multiple sub-jump-sets $ {\mathcal{J}}_i$ at a given time if, for instance, two or more measurements arrive at this instant of time (\ie, $t^i_k = t^j_{k'}$ with $i\neq j, k,k'\in \mathbb{N}$). In this case, for each $x\in \mathcal{J}$, define $\mathcal{I}(x): = \{i\in \mathbb{I}_{\geq 0}:  x\in \mathcal{J}_i\}$ as the nonempty index set. Then, one has $x\in \cup_{i\in \mathcal{I}(t)} \mathcal{J}_i$ and $x^+ \in \cup_{i\in \mathcal{I}(x)} G_i(x)$. Note that the flow set $\mathcal{F}$ and jump set $\mathcal{J}$ are both closed, and $\mathcal{F} \cup \mathcal{J} = \mathcal{S}$. Note also that with the introduction of the virtual timers $\tau_i$, the hybrid system $\mathcal{H}_o$ is autonomous and satisfies the hybrid basic conditions of \cite[Assumption 6.5]{goebel2012hybrid}.

	Consider the following Lyapunov function candidate: 
	\begin{align}
		V(x) & = V_R(x_o)  + \varepsilon \sum_{i=1}^N V_r^{i}(\tilde{r}_i,\tau_i) \label{eqn:def_V}
	\end{align}
	where $\varepsilon>0$. From \eqref{eqn:dot_V1} and \eqref{eqn:dotV_R}, the time-derivative of $V$ along the flows of \eqref{eqn:Hybrid-Closed-Loop} is given as
	\begin{align}
		\dot{V}(x) & \leq -2k_o \|\psi(A\mathcal{T})\|^2   - k_\theta  |\nu|^2   + 2  k_\theta \overline{c}_{\Gamma} |\nu| \|\tilde{r} \| \nonumber                             \\
		& \quad    + 2(2 k_\theta \overline{c}_{\Gamma}  +  \widehat{c}_\Gamma )\|\psi(A\mathcal{T}) \|  \|\tilde{r} \|  -\mu  \varepsilon \|\tilde{r}\|^2 \nonumber \\
		& \leq    -k_o \|\psi(A\mathcal{T})\|^2   - \frac{1}{2}k_\theta   \nu^2    - c_\varepsilon \|\tilde{r}\|^2  \label{eqn:dot_V}
	\end{align}
	where  $c_\varepsilon: = \mu  \varepsilon  - \frac{(2 k_\theta \overline{c}_{\Gamma}  +  \widehat{c}_\Gamma )^2}{k_o}- 2k_\theta \overline{c}_{\Gamma}^2$, and we used $\mathcal{T}=  \tilde{R}\mathcal{R}_u(\theta)$, $e^{\mu \tau_i} \geq 1$ for all $\tau_i\in [0,T_M^i]$, and the inequality $2b\|x\|\|y\| \leq a \|x\|^2 + \frac{b^2}{a} \|y\|^2$ for all $a, b>0, x,y\in \mathbb{R}^n$. To ensure $c_\varepsilon>0$, it is sufficient to choose the constant scalar $\varepsilon$ large enough such that $\varepsilon > \frac{1}{\mu} (\frac{(2 k_\theta \overline{c}_{\Gamma}  +  \widehat{c}_\Gamma )^2}{k_o} + 2k_\theta \overline{c}_{\Gamma}^2)$ with $\mu$ chosen as per Appendix \ref{sec:tilde_r}. Thus, $\dot{V}(x)$ is negative semi-definite and $V(x)$ in non-increasing along the flows of \eqref{eqn:Hybrid-Closed-Loop}.

	On the other hand, we consider the case where the state is in the jump set, \ie,  $x\in \mathcal{J}$. In this case, we have the following three scenarios:
	\begin{itemize}
		\item [1)] The scenario where only the jumps of the intermittent measurements are involved. One has $x\in \cup_{i\in \mathcal{I}(x)} \mathcal{J}_i$ with $\mathcal{I}(x) \subseteq \mathbb{I}$ and for each jump $x^+\in \cup_{i\in \mathcal{I}(x)} G_i(x)$, which implies that   $x_o^+ = x_o$ and
		\begin{align*}
			\begin{cases}
				\tilde{r}_i^+ = (1-k_r) \tilde{r}_i, \tau_i^+ \in [T_m, T_M], &  i \in \mathcal{I}(x)    \\
				\tilde{r}_i^+ =  \tilde{r}_i, \tau_i^+ = \tau_i,              &  i \notin \mathcal{I}(x)
			\end{cases}
		\end{align*}
		In view of \eqref{eqn:dot_V1}, \eqref{eqn:dotV_R} and \eqref{eqn:def_V}, one can show that
		\begin{align}
			V(x^+) - V(x) & = V_R(x_o^+)  + \varepsilon \sum_{j=1}^N V_r^{j}(\tilde{r}_j^+,\tau_j^+)  \nonumber                                   \\
			& \quad -   V_R(x_o)  - \varepsilon \sum_{j=1}^N V_r^{j}(\tilde{r}_j,\tau_j)  \nonumber                                 \\
			& \leq \varepsilon \sum_{i\in \mathcal{I}(x)}(V_r^{i}(\tilde{r}_i^+,\tau_i^+) -V_r^{i}(\tilde{r}_i,\tau_i) ) . \nonumber
		\end{align}
		From \eqref{eqn:V1^+}, one has $V_r^i(\tilde{r}_i^+,\tau_i^+) \leq \lambda_J V_r^i(\tilde{r}_i,\tau_i)$ with $\lambda_J = e^{\mu \bar{T}_M} (1-k_r)^2 < 1$. Hence, one can further show that
		\begin{align}
			V(x^+) - V(x)
			& \leq -\varepsilon (1-\lambda_J) \sum_{i\in \mathcal{I}(x)} V_r^{i}(\tilde{r}_i,\tau_i)    \label{eqn:V+1}.
		\end{align}
		Since $0<\lambda_J<1$ and $V_r^{i}(\tilde{r}_i,\tau_i)\geq 0$ for $i\in \mathbb{I}$, one can conclude that $V(x)$ is non-increasing after each jump at $x\in \cup_{i\in \mathcal{I}(x)} \mathcal{J}_i$ with $\mathcal{I}(x) \subseteq \mathbb{I}$.

		\item [2)]   The scenario where only the jump of $\theta$ is involved. One has $x\in  \mathcal{J}_0$ which implies that $x^+\in G_0(x)$, \ie, $x_o^+ \in G_o(x_o,\tilde{r})$ with $G_o(x_o,\tilde{r})$ defined in \eqref{eqn:def_Fo_Jo} and $\tilde{r}^+ = \tilde{r}, \tau^+ = \tau$. From  \eqref{eqn:V_R+}, \eqref{eqn:Hybrid-Closed-Loop} and \eqref{eqn:def_V} one can show that
		\begin{align}
			V(x^+) - V(x)
			& = V_R(G_o(x_o,\tilde{r}))  - V_R(x_o) \nonumber                             \\
			& = V_R(\tilde{R},\theta^+)  - V_R(\tilde{R},\theta) \nonumber                \\
			& \leq  - \delta   +  2 \sum_{i=1}^N \rho_i \|\tilde{r}_i\|^2 \label{eqn:V+2}
		\end{align}
		where $  \theta^+ \in  	\{{\theta}'\in \Theta: \mu_{\tilde{\phi}}(\tilde{R},\theta',\tilde{r})=0  \}$ as per \eqref{eqn:def_Fo_Jo}. Then, one has $V(x^+) \leq V(x)$ if $\sum_{i=1}^N \rho_i \|\tilde{r}_i\|^2 \leq \frac{\delta}{2}$.
		
		\item [3)] The scenario where both the jumps of the intermittent measurements and the switching variable $\theta$ are involved (\ie, the combination of above two scenarios, $x\in \cup_{i\in \mathcal{I}(x)} \mathcal{J}_i$ with $\mathcal{I}(x) \subseteq \mathbb{I}_{\geq 0}$ and $0\in \mathcal{I}(x)$). From \eqref{eqn:V+1} and \eqref{eqn:V+2}, one has
		\begin{align}
			V(x^+) - V(x)
			& \leq  -\varepsilon (1-\lambda_J) \sum_{i\in \mathcal{I}(x)\setminus \{0\}} V_r^{i}(\tilde{r}_i,\tau_i) \nonumber \\
			& \qquad  + V_R(G_o(x_o,\tilde{r}))  - V_R(x_o)  \nonumber                                                         \\
			& \leq  - \delta   +  2 \sum_{i=1}^N \rho_i \|\tilde{r}_i\|^2 \label{eqn:V+3}
		\end{align}
		which implies that $V(x^+) \leq V(x)$ if $\sum_{i=1}^N \rho_i \|\tilde{r}_i\|^2 \leq \frac{\delta}{2}$.
	\end{itemize}
 
	Next, we are going to show that the set $\mathcal{A}$ is stable. We first show that, for each constant $0<c_\delta<\frac{\delta}{2}$, there exist  a finite time $T_0 \geq 0$ and a finite number of jump $J_0\geq 0$ with $(T_0,J_0)\in \dom x$ such that $\sum_{i=1}^N \rho_i \|\tilde{r}_i(t,j)\|^2 \leq c_\delta$ for all $(t,j)\in \dom x, (t,j) \succeq (T_0,J_0)$ (\ie, $ t\geq T_0, j\geq J_0$). From \eqref{eqn:tilde_r_exp}, one can verify that, for each  $i\in \mathbb{I}$, $\|\tilde{r}_i(t,j)\|^2 \leq e^{\mu \bar{T}_M}  e^{- \lambda t}  \|\tilde{r}_i(0,0)\|^2 $ for all $(t,j)\in \dom x$.  Hence, for each $0<c_\delta<\frac{\delta}{2}$, there exists a finite time $T_0 \geq 0$ such that $\sum_{i=1}^N \rho_i \|\tilde{r}_i(t,j)\|^2 \leq c_\delta$ for all $(t,j)\in \dom x$ and $t \geq T_0$. Then, it remains to show that there is a finite number of jumps in the time interval $[0,T_0]$. 
	Note that the number of jumps in the set $\cup_{i\in \mathbb{I}}\mathcal{J}_i$ is finite (bounded by $T_0/\min\{T_m^1,\dots,T_m^N\}$) in the time interval $[0,T_0]$ by Assumption \ref{assum:intermittent}. 
	In view of \eqref{eqn:tilde_r_exp}, \eqref{eqn:def_V}-\eqref{eqn:V+1} and the fact $V_R(G_o(x_o)) = \tr((I_3-\mathcal{T}(\tilde{R},\theta^+)A) + \frac{\gamma}{2} (\theta^+)^2 \leq 4\lambda_M^A + \frac{\gamma}{2} \max_{\theta'\in \Theta} |\theta'|^2$ with $\lambda_M^A$ denoting the maximum eigenvalue of $A$, one can verify that $V(x(t,j))$ is   upper bounded for all $(t,j)\in \dom x$. Using the following inequalities modified from \eqref{eqn:phi_V_R}:
	\begin{align}
		\tilde{\phi}(x_o,\tilde{r})
		& \leq V_R(x_o) + \sum_{i=1}^N \rho_i \|\tilde{r}_i\|^2 \nonumber                                                  \\
		& \leq V_R(x_o) +  \varepsilon \sum_{i=1}^N \frac{\rho_i}{\varepsilon} e^{\mu \tau_i}\|\tilde{r}_i\|^2   \nonumber \\
		& \leq c_{\phi} V(x)
	\end{align}
	with $c_\phi: = \max\{1,\frac{\rho_1}{\varepsilon},\dots, \frac{\rho_N}{\varepsilon}\}$, one can also show that $\tilde{\phi}(x_o,\tilde{r})$ is upper bounded for all $(t,j)\in \dom x$.	Moreover, since $\tilde{\phi}(x_o,\tilde{r})$ is non-negative and has a strict decrease after each jump in the set $\mathcal{J}_0$ by the definition of the jump set $\mathcal{J}_0$  in \eqref{eqn:def_Fo_Jo}, the number of consecutive jumps in the set $\mathcal{J}_0$ is finite (\ie, no Zeno behavior). 
	Hence, the number of jumps in the jump set $\mathcal{J}$ is finite and there is no finite escape-time in the time interval $[0,T_0]$, \ie, there exists a finite number of jumps $J_0$ such that $(T_0,J_0)\in \dom x$ and $\sum_{i=1}^N \rho_i \|\tilde{r}_i(t,j)\|^2 \leq c_\delta$ for all $(t,j)\in \dom x, (t,j) \succeq (T_0,J_0)$.	Then, for all $(t,j) \in \dom x, (t,j) \succeq (T_0,J_0)$, one obtains from \eqref{eqn:V_R+} that
	\begin{align}
		 V_R(G_o(x_o,\tilde{r}))  - V_R(x_o)  
		&   \leq  - \delta   +  2 \sum_{i=1}^N \rho_i \|\tilde{r}_i\|^2 \nonumber \\
            &   \leq 2 c_\delta - \delta < 0  \label{eqn:V_R_delta}.
	\end{align} 
	This implies that, after each jump in the set $\mathcal{J}_0$, $V(x)$ is bounded for all $(t,j) \in \dom x, (t,j) \preceq (T_0,J_0)$ and has a strict decrease  for all $(t,j) \in \dom x, (t,j) \succeq (T_0,J_0)$.
	Substituting \eqref{eqn:V_R_delta} into \eqref{eqn:V+2} and \eqref{eqn:V+3} and using the results in \eqref{eqn:def_V} and  \eqref{eqn:V+1}, one can show that  $V(x)$ is bounded for all $(t,j) \in \dom x, (t,j) \preceq (T_0,J_0)$ and non-increasing for all $(t,j) \in \dom x, (t,j) \succeq (T_0,J_0)$.  
	It is important to point out that $T_0=J_0=0$ if the initial conditions are small enough (\ie, $\sum_{i=1}^N \rho_i \|\tilde{r}_i(0,0)\|^2  <\frac{\delta}{2}$).
	Therefore, one can conclude that the set $\mathcal{A}$ is stable, and any maximal solution to \eqref{eqn:Hybrid-Closed-Loop} is bounded.

	Now, we are going to show that the set $\mathcal{A}$ is globally attractive, \ie,  every maximal solution $x$ to   \eqref{eqn:Hybrid-Closed-Loop}  is complete and satisfies $\lim_{t+j\to \infty}|x(t,j)|_{\mathcal{A}} = 0$ for all $(t,j)\in \dom x$. 
	Define the following nonempty set:
	\begin{equation}
		\mathcal{U}: = \left\{x= (x_o, \tilde{r},\tau)\in \mathcal{S} \left| \sum_{i=1}^N \rho_i \|\tilde{r}_i\|^2  < c_\delta \right. \right\}
	\end{equation}
	with some $0<c_\delta<\frac{\delta}{2}$ and $(T_0,J_0) \in \dom x$. Note that the closed-loop system \eqref{eqn:Hybrid-Closed-Loop} satisfies the hybrid basic condition \cite[Assumption 6.5]{goebel2012hybrid}, $F(x)\subset T_{\mathcal{F}}(x)$ for any $x \in \mathcal{F}\setminus \mathcal{J}$  with $T_{\mathcal{F}}(x)$ denoting the tangent cone to $\mathcal{F}$ at point $x$, $G(\mathcal{J}) \subset \mathcal{F}\cup \mathcal{J} = \mathcal{S}$, and every maximal solution to \eqref{eqn:Hybrid-Closed-Loop} is bounded. Therefore, by virtue of \cite[Proposition 6.10]{goebel2012hybrid}, every maximal solution to \eqref{eqn:Hybrid-Closed-Loop} is complete. Thus, every maximal solution $x$, since complete and bounded,  is a compact hybrid trajectory satisfying  
	$$
	\overline{\{x(t,j)| (t,j)\in \dom x, (t,j) \succeq (T_0,J_0)\}} \subset \mathcal{U}.
	$$
	
	We consider the following functions motivated from \cite{sanfelice2007invariance}
	\begin{align}
		u_C(x) := \begin{cases}
			-k_o \|\psi(A\mathcal{T})\|^2   - \frac{1}{2}k_\theta   \nu^2    - c_\varepsilon \|\tilde{r}\|^2 , &  x\in \mathcal{F} \\
			-\infty,                                                                                           & \text{otherwise}
		\end{cases} \label{eqn:def_u_C}
	\end{align}
	and
	\begin{align}
		u_D(x) := \begin{cases}
			\max_{i\in \mathcal{I}(x)} \{V(G_i(x)) - V(x)\}, &  x\in \mathcal{J} \\
			-\infty,                                         & \text{otherwise}
		\end{cases} \label{eqn:def_u_D}
	\end{align}
	with $\mathcal{I}(x): = \{i\in \mathbb{I}_{\geq 0}: x\in \mathcal{J}_i\}$. It is clear from \eqref{eqn:dot_V} that $u_C(x) \leq 0$ for all $x\in \mathcal{U}$, and from \eqref{eqn:V+1}-\eqref{eqn:V+3} and \eqref{eqn:V_R_delta}  that $u_D(x)\leq 0$ for all $x\in \mathcal{U}$.  
	By virtue of \cite[Theorem 4.7]{sanfelice2007invariance}, every maximal solution $x$ to the hybrid  system \eqref{eqn:Hybrid-Closed-Loop}  must converge to the largest weakly invariant set contained in
	\begin{equation*}
		V^{-1}(r) \cap \mathcal{U} \cap [u_C^{-1}(0) \cup (u_D^{-1}(0) \cap G(u_D^{-1}(0)))]
	\end{equation*}
	for some constant $r\in V(\mathcal{U})$. Applying simple set-theoretic arguments, one obtains $V^{-1}(r) \cap \mathcal{U} \cap [u_C^{-1}(0) \cup (u_D^{-1}(0) \cap G(u_D^{-1}(0)))] =  [V^{-1}(r) \cap \mathcal{U} \cap 	u_C^{-1}(0)] \cup [V^{-1}(r) \cap \mathcal{U} \cap \left( u_D^{-1}(0) \cap  G(u_D^{-1}(0))\right)]$. From \eqref{eqn:V+1}-\eqref{eqn:V+3} and the property of the jump set $\mathcal{J}_0$ (\ie, $\mathcal{J}_o$) in \eqref{eqn:V_R_delta}, one can show that $ 
	u_D^{-1}(0) \cap  G(u_D^{-1}(0))= \emptyset
	$. Moreover, from \eqref{eqn:def_u_C}, one has
	\begin{align*}
		& V^{-1}(r) \cap \mathcal{U} \cap 	u_C^{-1}(0)                                                          \\
		& = V^{-1}(r) \cap \mathcal{U} \cap  \{x \in \mathcal{F}:\psi(A\mathcal{T})=0,\tilde{r}=0, \nu   = 0\} \\
		& = \{x\in \mathcal{F}:\psi(A\mathcal{T})=0,\tilde{r}=0,   \theta  = 0\}                               \\
		& = \{x \in \mathcal{F}:\psi(A\tilde{R})=0,\tilde{r}=0,   \theta  = 0\}
	\end{align*}
	where we used the facts $\nu = \gamma  \theta  + 2   u\T\psi(A\mathcal{T}) $, $\mathcal{T} = \tilde{R}\mathcal{R}_u(\theta)$. Since $\mathcal{F} =  \mathcal{F}_o  \times [0,T_M^1] \times \cdots \times [0,T_M^N] =   \{x\in \mathcal{S} : \mu_{\tilde{\phi}}(\tilde{R},\theta,\tilde{r}) \leq \delta \}$, one can conclude from Lemma \ref{lem:Psi} and the definition of the set $\mathcal{A}$ that
	\begin{align*}
		& V^{-1}(r) \cap \mathcal{U} \cap 	u_C^{-1}(0)                                                                                     \\
		& = \{x \in \mathcal{S}:\mu_{\tilde{\phi}}(\tilde{R},\tilde{r},\theta) \leq \delta, \psi(A\tilde{R})=0, \theta =0, \tilde{r}=0 \} \\
		& = \{(\tilde{R},\theta,\tilde{r},\tau) \in \mathcal{S}: \tilde{R}=I_3, \theta=0,\tilde{r}=0 \}                   \\
        &= \mathcal{A}.
	\end{align*} 
	Therefore, one can conclude that every maximal solution, since complete and bounded, must converge to the set $\mathcal{A}$, which implies that the set $\mathcal{A}$ is globally asymptotically stable for the hybrid system \eqref{eqn:Hybrid-Closed-Loop}. This completes the proof.

	\bibliographystyle{IEEEtran}
	\bibliography{mybib} 
	
\end{document}